\documentclass[onefignum,onetabnum]{siamonline250211}

\usepackage{lipsum}
\usepackage{amsfonts}
\usepackage{graphicx}
\usepackage{pgfplots}
\usepackage{subcaption}
\usepackage{epstopdf}
\usepackage{algorithmic}
\ifpdf
  \DeclareGraphicsExtensions{.eps,.pdf,.png,.jpg}
\else
  \DeclareGraphicsExtensions{.eps}
\fi

\usepackage{todonotes}

\usepackage[normalem]{ulem} 
\usepackage{xcolor}

 \newif\ifchanges
 \changestrue      
 \changesfalse   

\usepackage{cleveref}

\usepackage{array}        
\usepackage{booktabs}   
\usepackage{wrapfig}

\usepackage{enumitem}
\usepackage{graphicx}
\usepackage[percent]{overpic}  

\setlist[enumerate]{leftmargin=.5in}
\setlist[itemize]{leftmargin=.5in}

\newcommand{\bydef}{\,\stackrel{\mbox{\tiny\textnormal{\raisebox{0ex}[0ex][0ex]{def}}}}{=}\,}

\newcommand{\R}{\mathbb{R}}
\newcommand{\bond}{\mathrm{b}}
\renewcommand{\angle}{\mathrm{a}}
\newcommand{\Eharmonic}{E_{\mathrm{h}}}
\newcommand{\ETersoff}{E_{\mathrm{T}}}
\newcommand{\EHuber}{E_{\mathrm{H}}}

\newsiamremark{remark}{Remark}
\newsiamremark{hypothesis}{Hypothesis}
\crefname{hypothesis}{Hypothesis}{Hypotheses}
\newsiamthm{claim}{Claim}
\newsiamremark{fact}{Fact}
\crefname{fact}{Fact}{Facts}

\headers{Computer-Assisted Proofs for Geometry Optimization}{M. Ayala, R. Choksi, and B. Wirth}

  \title{Computer-Assisted Proofs for Geometry Optimization: \\
From Crystallization to Carbon Nanotubes
  }

\author{Miguel Ayala\thanks{ 
  Department of Mathematics and Statistics, McGill University, Montr\'eal, Canada, \email{miguel.ayala@mail.mcgill.ca}}
\and Rustum Choksi\thanks{Department of Mathematics and Statistics, McGill University, Montr\'eal, Canada, \email{rustum.choksi@mcgill.ca}}
\and Benedikt Wirth\thanks{Applied Mathematics M\"unster, University of M\"unster, M\"unster, Germany, \email{benedikt.wirth@uni-muenster.de}}}

\usepackage{amsopn}

\ifpdf
\hypersetup{
  pdftitle={Computer-assisted Proofs for Geometry Optimization},
  pdfauthor={Miguel Ayala}
}
\fi

\begin{document}

\maketitle

\begin{abstract}
We present a framework based on computer-assisted proofs that turns 
geometry optimization simulations for atomistic structures into mathematical
proofs. Starting from a
numerically computed approximation of a local minimizer or saddle point, we use
validated numerical computations to prove the existence of a critical point of
the potential energy close to this approximation. We demonstrate this framework in two settings. In the first, we
study capped carbon nanotubes modeled as minimizers of carbon interatomic
potentials (harmonic, Tersoff, and a Huber potential) and obtain proven
bounds on tube diameter, bond lengths, and bond angles. In particular, we show
that caps induce diameter oscillations along the tube. As a second
application, we consider a finite Lennard--Jones crystal in a face-centered
cubic (fcc) lattice and provide computer-assisted
proofs of a local minimizer
representing the perfect crystal, a local minimizer with a single vacancy
defect, and a saddle point that connects two single-vacancy configurations on
the energy landscape.
\end{abstract}



\section{Introduction}

Many molecular and atomistic systems, ranging from crystalline solids to carbon nanotubes, can be described as critical points of high-dimensional nonlinear energy functionals. The identification of low-energy critical points—commonly referred to as \textit{geometry optimization}—relies on numerical methods to locate local minima and the transition states that connect them. Advances in the development of interatomic potentials, together with increased computational resources, have established geometry optimization as a widely used approach across physics, chemistry, and the mathematical sciences.

While numerically found minimizers often provide both physical intuition and reliable information about the structures under study, there is no guarantee that a numerical approximation is close to an actual solution. 
We present a framework, based on computer-assisted proofs (validated numerical computations) \cite{Koch1996, Tucker2011, vanDenBerg2015, Nakao2019}, to obtain mathematical proofs of existence and stability for critical points of an energy functional from numerical approximations. These {\it proofs} can, in particular, validate differences between models and structures which might otherwise easily be attributed to nonconverged numerics or similar effects. Furthermore, they can mitigate numerical artifacts and spurious solutions caused by truncation, round-off, and implementation errors. The theory behind validated numerical computations is well-studied but, to our knowledge, this is the first application to geometry optimization for atomistic structures. 
In this paper, we focus on two classes of structures as test cases: a fairly detailed treatment of capped carbon nanotubes and a few examples for fcc crystallization. Let us first briefly discuss the motivation and scope of this article.

Our main motivation stems from questions in the mathematical community. The exact nature of minimizers of interaction potentials in mathematical physics and applied mathematics is a subject of great interest. 
Indeed, there is considerable interest and effort to validate, e.g.,
the stability of carbon nanotubes \cite{Friedrich_2019,Mainini_2017,MaMuPiSt16} or of various types of crystal structures \cite{MR4677219,MR2898772,MR3199992,MR3117623,MR3382763,MR3360741,MR2200888}. It is in the nature of things that such validation is lagging far behind physical experimental knowledge,
the above analyses all being restricted to highly simplified models
(exploiting for instance high symmetries, infinite extension, reduced dimensions, simplified atom interactions, single atom types, etc.). Less symmetric settings such as capped nanotubes or crystal defects and small size effects are currently out of reach with classical methods. Computer-assisted proofs here represent the natural next step to overcome these limitations and thereby provide new tools to the mathematical community.

On the other hand, we believe validated numerics for geometry optimization could also be relevant and useful to computational chemists and physicists (cf. \cite{Mehta_2015, Mehta_2014}). In most cases, interaction potentials lead to highly nonconvex energy landscapes, where the number of local minimizers grows at an exponential rate as the number of particles increases \cite{wales_2003, PhysRevE.59.48}. This large pool of candidates can make it easy to misinterpret numerical artifacts or nonconverged solutions as true local minimizers. This can in turn lead to incorrect conclusions about the predicted structures. In this context, our methodology allows us to measure geometric and energetic properties of minimizers with guaranteed certainty up to a very small error. The magnitude of these {\it error bars}, stemming from the theory (cf. \Cref{section:CAPs}), may be exceptionally small, providing a robust way to analyze numerical simulations for a given model. 

Is this useful for the practitioner? This depends on the potential and the minimizing state. While, mathematically, one can rigorously confirm geometric differences between critical points of the energy, it can often be the case that these differences are so small, relative to physically measurable parameters, that they are not physically meaningful. However, there are cases where they are meaningful; for example, in our carbon nanotube examples, the {\it error bars} produced by our proofs confirm that a model predicts two distinct bond lengths. To this end, we believe validated numerics have the potential to be useful in addressing certain questions of interest in computational science.

The paper is organized as follows. We begin in \Cref{section:CAPs} by introducing our computer-assisted proof (CAP) method for validating critical points of energy functions that are invariant under translations and rotations, as is typical for atomistic systems in free space. We state a validation theorem that gives sufficient 
conditions to prove the existence and stability of a critical point near a numerical 
approximation and explain how we use a computer to verify these conditions.

As a first application in \Cref{section:geometries}, we study the geometric structure of carbon nanotubes. From a
mathematical point of view, the geometry of carbon nanotubes has also been
studied using reduced models based on geometric Ansätze
\cite{MaMuPiSt16, Mainini_2017, Friedrich_2019}. However, for carbon nanotubes,
except for such reduced and simplified models, mathematical proofs about the
structure of local minimizers are scarce. This makes carbon nanotubes a natural test case for the 
application of our method. Moreover, experimental and theoretical studies indicate that geometric features such as diameter, chirality, bond lengths, bond angles, and the way the tubes are capped influence the mechanical and electronic properties of carbon nanotubes \cite{dresselhaus2001carbon, meyyappan2004carbon, Kanamitsu2002, MaMuPiSt16, Favata2016, Odom2000}. For a given model, our computer-assisted approach provides a way to rigorously study these geometric differences.

We apply our method to several nanotube geometries. As a consequence of the computer-assisted proof, we obtain rigorous bounds (error bars) on the coordinates of each atom. These bounds let us demonstrate geometric differences between tubes and derive rigorous bounds for quantities that depend on the coordinates, such as the energy. 
Using the harmonic interatomic potential \eqref{eq:energy_harmonic}, we study how different chiralities and cap configurations change the geometric features of the tube. We then use the well-known empirical Tersoff potential \eqref{eq:energy_tersoff} 
\cite{PhysRevB.37.6991,PhysRevLett.61.2879,Marchant2023} to examine how the choice of 
interaction potential influences diameters, bond lengths, and bond angles. With 
this potential, we rigorously count the number of distinct bond lengths and bond 
angles in a local minimizer, and we compare these results with those obtained 
from a mathematically motivated potential \eqref{eq:energy_epsilon} designed to 
study the role of differentiability in the geometry of the tube. 

Our method applies broadly to other geometries and interaction potentials. For example, the carbon nanotube analysis can be applied to other carbon allotropes such as 
fullertubes~\cite{Koenig_2020,Stevenson_2024} and 
cyclocarbons~\cite{Kaiser_2019}, and to more complex macromolecules, as long 
as the interaction potential is given by an analytic or piecewise-analytic 
expression. We believe the examples presented, particularly for carbon nanotubes with Tersoff bond-order potentials, serve as a proof of concept for this approach. In particular, they open a path toward adoption for more complicated modern potentials. 

Our second application, discussed in \Cref{section:lattice}, stems from crystallization. Lennard--Jones interacting systems are a standard model in the mathematical study of crystallization, 
where one expects low-energy configurations of many atoms to arrange in 
crystalline lattices, for example the face--centered cubic (fcc) lattice \cite{BlancLewin2015}. We consider \(n\) atoms interacting through the Lennard--Jones potential and study geometry 
optimization for a finite fcc lattice configuration. The field presents formidable challenges for rigorous mathematics with only a few results in the literature (see for example \cite{MR2200888, MR3360741}). We study a numerically computed local minimizer that models a stable fcc crystal, and we use our framework to turn this approximation into a mathematical proof.
The result shows that the configuration is a stable local minimizer, and that it remains locally stable after introducing a single vacancy defect. We also prove the existence of a saddle point connecting two defective crystals with neighboring vacancies. We include fcc examples with open boundary conditions to illustrate our approach in the setting of finite particle clusters. Finite particle clusters have been studied extensively, for example in the context of Lennard--Jones clusters and the Thomson problem \cite{Wales_97, Mainini_2017, Wales_2004}. We end the section with an example for periodic boundary conditions, where a repeating cell represents an infinite crystal.

We conclude the paper with a brief discussion of the limitations of the approach and of the numerical and physical challenges involved in the broader integration of computer-assisted proofs into the contemporary use of geometry optimization for atomistic structures.

\section{Computer-assisted Proofs for Geometry Optimization}
 \label{section:CAPs}
We consider a system of \( n \) points in $\mathbb{R}^3$, where each point represents an atom. The positions of these atoms are given by
\[
    p_i = (x_i, y_i, z_i) \in \mathbb{R}^3, \quad i = 1, \dots, n.
\]
The full set of \( n \) atoms, the atom cloud, is denoted as \( p \in \mathbb{R}^{3n} \), where
\[
    p \bydef  (x_1, \dots, x_n, y_1, \dots, y_n, z_1, \dots, z_n).
\]
In this paper, we compute and prove the existence of critical points of an
energy function \(E\) modeling interactions between atoms. These points
are the solutions of the zero-finding problem
\[
    F(p) \bydef \nabla E(p) = 0.
\]
In particular, we consider energy functions
\(
E: \mathbb{R}^{3n} \to \mathbb{R}
\)
that depend only on the distances between atoms and the angles formed between pairs of atoms. Since rotations and translations leave distances and angles unchanged, the energy functions remain invariant under these transformations. 
In this work we focus on two types of critical points: local minimizers that
correspond to low-energy geometries and saddle points that correspond to
transition states. Because this problem is posed in a high-dimensional space,
practitioners in geometry optimization rely on numerical simulations to compute
and study solutions. 

Our central question is: To what extent can computer simulations rigorously
describe the geometry of a critical point given a potential energy function? To
address this, we use computer-assisted proofs, also called validated numerics,
which produce results at the level of mathematical proof
\cite{Koch1996, Tucker2011, vanDenBerg2015, Nakao2019}.

Suppose \(\bar{p} \in \mathbb{R}^{3n}\) is an approximation to a zero of the gradient of some energy function \(E\). We want to show, using a fixed-point argument, that a unique true zero of \(F\) exists within a ball centered at \(\bar{p}\) for some radius \(r\). In particular, we will use the following result, a Newton-Kantorovich-type theorem for finite-dimensional problems. This theorem is related to the Krawczyk operator approach \cite{MR0255046} and the interval Newton method \cite{MR0231516}. In the following, $\|\cdot\|$ denotes both a norm on $\mathbb R^m$ as well as the induced matrix norm, while $B_r(\bar x)$ denotes the open norm-ball around $\bar x$ of radius $r$, and $\overline B_r(\bar x)$ denotes the corresponding closed norm-ball.
\begin{theorem}[Validated zero]
    \label{theorem:validation}
    Let $f:\R^{m} \to \R^{m} $ be differentiable and $\bar{x} \in \R^{m}$, $A\in\R^{m\times m}$. Suppose that for some $r^{*}>0$ and $Y,Z(r^*)>0$ we have
    \[
        \left \|  Af(\bar{x})  \right \| \leq Y,
        \quad
        \sup_{x \in \overline B_{r^{*}}(\bar{x})} \left \|I - ADf(x) \right \| \leq Z(r^{*})<1,
    \]
    \begin{equation*}
                 0< r_{\min} \bydef \frac{Y}{1-Z(r^{*})}< r^{*}.
    \end{equation*}        
    Then for any $r$ in the interval $[r_{\min}, r^{*}]$ there exists a unique zero of $f$ in $\overline {B}_r(\bar{x})$. 
\end{theorem}
\begin{proof}
    Suppose $r_{\min} <r\leq r^*$. We will show that the Newton-like operator \(T:\mathbb{R}^{m} \to \mathbb{R}^{m}\),
    \begin{equation*}
        T(x) \bydef x - A f(x),
    \end{equation*}
    has a unique fixed point in the ball $\overline B_{r}(\bar x)$ centered at the approximate zero \(\bar{x}\). 
    First, we prove that \(T\) is a contraction. 
    For any \(x, y \in \overline{B}_{r}(\bar{x})\) we have
    \[
        \lVert T(x) - T(y)\rVert \leq \sup_{z \in \overline {B}_{r}(\bar{x})} \lVert DT(z) \rVert
        \lVert x -y\rVert \leq Z(r^*)\lVert x-y\rVert.
    \]
    Next, we show that \(T\) maps \(\overline B_{r}(\bar{x})\) into itself. Let $x\in\overline B_{r}(\bar{x})$ and observe that
    \begin{multline*}
        \lVert T(x) - \bar{x}\rVert \leq 
        \lVert T(x) - T(\bar{x})\rVert + \lVert T(\bar{x}) - \bar{x}\rVert \\
        \leq \sup_{z \in \overline {B}_{r}(\bar{x})}  \lVert DT(z) \rVert \lVert x - \bar{x} \rVert + Y
        \leq Z(r^*) \lVert x - \bar{x} \rVert + Y
        < r.
    \end{multline*}
    By the Contraction Mapping Theorem \cite{Chicone2006}, \(T\) has a unique fixed point in \(\overline{B}_{r}(\bar{x})\). 
    Finally,
    observe that the matrix \(A\) is invertible.
    Indeed, otherwise, pick a unit-norm vector $v$ from the kernel of $ADf(\bar x)$; then $1=\|v\|=\|(I-ADf(\bar x))v\|\leq Z(r^*)\|v\|=Z(r^*)<1$ yields a contradiction.
    The invertibility of $A$ implies that \(f\) has a true zero in \(\bar{B}_{r}(\bar{x})\).
\end{proof}
The formulation we use is standard in validated numerics and has been applied recently to finite-dimensional problems, for example in the study of relative equilibria in the \(N\)-body problem \cite{MR4336016,MR4028598,MR3917433}. We include the short proof above for completeness. 

In practice, $\bar{x}$ will be a numerical approximation of a zero of $f$, and $A$ will be a numerical approximation to the inverse of $Df(\bar{x})$. Furthermore, we will employ the supremum norm $\|\cdot\|=\|\cdot\|_\infty$ since then the required estimates can be rigorously verified via classical interval arithmetic. In what follows we say that we have \textit{validated} the approximate zero \( \bar{x} \)
if, by applying \Cref{theorem:validation}, we prove the existence of a true zero of \(f\) in a neighborhood of \( \bar{x} \).

In our case, we cannot apply \Cref{theorem:validation} directly to find zeros
of $\nabla E$ because the energy function $E$ has continuous invariances. Indeed, $E$ is invariant under rotations and translations so that the zeros of $\nabla E$ form a continuous six-dimensional set, obtained by rotating or translating any critical point. \Cref{theorem:validation}, in contrast, provides conditions guaranteeing a unique, isolated zero within a small ball.

To solve this problem, we extend the system by six additional equations that enforce orthogonality to the invariant directions near a (given, but arbitrary) numerical approximation \( \bar{p} = (\bar{x}, \bar{y}, \bar{z}) \in \mathbb{R}^{3n} \) of a critical point. To make the system square again, we further introduce six additional variables that act like Lagrange multipliers to the new constraints. The extended zero-finding problem
\(
    F_{\bar{p}}: \mathbb{R}^{3} \times \mathbb{R}^{3} \times \mathbb{R}^{3n} \to \mathbb{R}^{3} \times \mathbb{R}^{3} \times \mathbb{R}^{3n}
\)
is given by
\begin{equation}
    \label{eq:validation_map}
    F_{\bar{p}}
    (
    \tau
    ,
    \rho
    ,
    p
    )
    \bydef
    \begin{bmatrix}
        T_{1} \cdot  \left( \bar{p} - p \right)
        \\
        T_{2} \cdot \left( \bar{p} - p \right)
        \\
        T_{3} \cdot \left( \bar{p} - p \right)
        \\
        p^TR_{1}\bar{p}
        \\
        p^TR_{2}\bar{p}
        \\
        p^TR_{3}\bar{p}
        \\
        \nabla E(p)
        +
        S(\tau, \rho, p)
    \end{bmatrix},
\end{equation}
where $R_i\in\mathbb{R}^{3n\times 3n}$  for $i = 1,2,3$ denote the infinitesimal generators of rotations in 3D
\[
    R_1 = \begin{pmatrix}
        0 & -I & 0 \\
        I & 0  & 0 \\
        0 & 0  & 0
    \end{pmatrix},
    \quad
    R_2 = \begin{pmatrix}
        0 & 0 & 0  \\
        0 & 0 & -I \\
        0 & I & 0
    \end{pmatrix}
    \quad and
    \quad
    R_3 = \begin{pmatrix}
        0  & 0 & I \\
        0  & 0 & 0 \\
        -I & 0 & 0
    \end{pmatrix}
\]
(with $I\in\R^{n\times n}$ the $n\times n$ identity matrix) and $T_i \in \mathbb{R}^{3n}$ for $i=1,2,3$ denote the infinitesimal generators of translations
\[
    [T_1 ]_{i} =
    \left\{\begin{matrix}
        1 &  \text{if } 0\leq i \leq n \\
        0 & \text{otherwise}      \\
    \end{matrix}\right.,
    \quad
    [T_2 ]_{i} =
    \left\{\begin{matrix}
        1 & \text{if } n+1 \leq i \leq 2n \\
        0 & \text{otherwise}         \\
    \end{matrix}\right.
    ,\quad
    [T_3 ]_{i} =
    \left\{\begin{matrix}
        1 &  \text{if } 2n+1 \leq i \leq 3n \\
        0 & \text{otherwise}         \\
    \end{matrix}\right. .
\]
Finally, the function $S:  \mathbb{R}^{3} \times \mathbb{R}^{3} \times \mathbb{R}^{3n} \to  \mathbb{R}^{3n}$ is defined as

\[
    S(\tau, \rho, p) =
    \sum_{i=1}^{3} \tau_i T_i + \sum_{i=1}^{3} \rho_i R_i p.
\]

The following result shows the equivalence between the zeros of \( F_{\bar{p}} \) and of $\nabla E$.
\begin{theorem}[Removal of invariances]
    \label{theorem:equivalence}
    Suppose \( F_{\bar{p}}(\tau,\rho,p) = 0 \). Then also \( \nabla E(p) = 0 \).
    \begin{proof}
        The gradient of the energy $E$ is orthogonal to the invariant directions. That is, for each symmetry generator \( i = 1,2,3 \) we have
        \[
            \nabla E(p)^TR_{i}p = 0, \quad \nabla E(p)^TT_{i} = 0.
        \]
Hence we have

\begin{equation}
\label{eq:gradE_dot_S}
\nabla E(p) \cdot S(\tau, \rho, p) =
\nabla E(p) \cdot \left( \sum_{i=1}^{3} \tau_i T_i + \sum_{i=1}^{3} \rho_i R_i p \right) = 0,
\end{equation}

        which implies
        \[
            0 = \left[ \nabla E(p) + S(\tau, \rho, p) \right] \cdot S(\tau, \rho, p) = S(\tau, \rho, p) \cdot S(\tau, \rho, p) = \lVert S(\tau, \rho, p) \rVert_{2}^2
        \]
        and thus \( S(\tau, \rho, p) = 0 \). Therefore, \( \nabla E(p) = 0 \).
    \end{proof}
\end{theorem}
A similar numerical strategy appears in \cite{doi:10.1137/19M1269865,MUNOZALMARAZ20031} to study periodic orbits in Hamiltonian systems. In the field of computer-assisted proofs, the technique of adding variables and equations to remove continuous invariances from a zero-finding problem is also known as introducing unfolding parameters (e.g.\ \cite{calleja2024}). Theorem~\ref{theorem:equivalence} removes the degeneracies at critical points caused by continuous symmetries. In contrast, discrete symmetries such as permutations require no special treatment. They do not generate continuous families of critical points and therefore do not lead to zero eigenvalues in the Hessian of the energy.

Once an approximate local minimizer has been validated as a critical point of
the energy, we use the following result to prove that it is indeed a local
minimum.

\begin{corollary}[Validated local minimizers of $E$]\label{theorem:minimizer}
Let $\bar x=(\bar \tau,\bar \rho,\bar p),r^*,r_{\min},Y,Z(r^*),A$ as in \Cref{theorem:validation} for $f=F_{\bar{p}}$ and $\|\cdot\|=\|\cdot\|_\infty$, and abbreviate by $M_i(p)\in\R^{i\times i}$ the upper left submatrix of

\[
H(p) \bydef \nabla^2 E(p) + \sum_{i=1}^{3} T_i T_i^{\top} + \sum_{i=1}^{3} (R_i p)(R_i p)^{\top}.
\]

If $(M_i^{-1}(p))_{ii}>0$ for all $p\in B_{r_{\min}}(\bar{x})$ and $i=1,\ldots,3n$, then $E$ has a unique (up to continuous invariances) local minimum in $B_r(\bar p)$ for all $r\in[r_{\min},r^*]$.
\end{corollary}
\begin{proof}
\Cref{theorem:validation,theorem:equivalence} already imply the existence of a critical point of $E$ in $B_r(\bar p)$, so it remains to show the minimizing property.
To this end it suffices to show positive semidefiniteness of $\nabla^2 E(p)$ for all $p\in B_r(\bar p)$.
This in turn is implied by positive definiteness of $H(p)$, since $T_1,T_2,T_3,R_1p,R_2p,R_3p$ (the invariant directions of energy $E$) lie in the kernel of $\nabla^2E(p)$.
By Sylvester's criterion, \( H(p) \) is positive definite if and only if all its leading principal minors $\Delta_i\bydef\det M_i(p)$, $i=1,\ldots,3n$, are positive.
Using the cofactor expansion for the inverse of \( M_i(p) \), we have
\[
M_i(p)^{-1} = \frac{1}{\Delta_i} \operatorname{cof}(M_i(p))^\top
\]
and in particular
\[
\frac{\Delta_{i-1}}{\Delta_i}=[M_i(p)^{-1}]_{ii} >0
\]
for the \((i,i)\)-entry of \( M_i^{-1}(p) \).
With \( \Delta_0 = 1 > 0 \) it follows by induction that \( \Delta_i > 0 \) for all $i$, as desired.
\end{proof}

There exist several alternative methods to check the positivity of a matrix such as~\cite{Rump2006}. Another possibility to obtain bounds for the eigenvalues is the Gershgorin circle theorem \cite{GolubVanLoan2013}, which provides intervals enclosing the spectrum (as in \Cref{theorem:lattice saddle}). For symmetric matrices, a more refined bound can be obtained by methods that validate each eigenvalue individually using a Newton operator argument as presented in~\cite{Rump2001}. 

To conclude this section we describe our computer-assisted technique, the practical application of \Cref{theorem:minimizer}, explicitly. We begin by computing an approximate minimizer $\bar{p}$ of $E$ using the BFGS quasi-Newton algorithm, terminating at six digits of accuracy. We then refine $\bar{p}$ by applying Newton's method to the function $F_{\bar{p}}$ from \eqref{eq:validation_map} with initialization $(0,0,\bar{p})$, yielding an approximate zero $\bar x=(\bar\tau,\bar\rho,\bar p)$ of $F_{\bar{p}}$ with ten to twelve digits of accuracy.

To find $r^*,Y,Z(r^*)$ verifying the conditions of \Cref{theorem:minimizer}, we use a computer in the following way. First, we compute matrix $A$ from \Cref{theorem:validation} as the numerical inverse of $DF_{\bar{p}}(\bar{x})$. We implement the function $F_{\bar{p}}$ and its Jacobian so that they can be evaluated using intervals instead of single floating-point numbers \cite{MR0231516}. More specifically, we use interval arithmetic to control round-off errors and work with basic functions (such as sums, products, and standard elementary functions) whose implementations are rigorous in the sense that they return intervals guaranteed to contain the exact result. 

We call an evaluation of a mathematical function \emph{rigorous} when, using
interval arithmetic, we compute an interval that is guaranteed to contain the
exact real value of the function at the given input. In particular, a rigorous bound \( Y \) is computed by evaluating
\[
    \left\| AF_{\bar{p}}(\bar{x}) \right\|_{\infty}
\]
in interval arithmetic as follows. Starting from the numerical approximation \( \bar{x} \), we first enclose it in
an interval vector whose entries are the smallest intervals with floating-point endpoints that contain the corresponding components of \( \bar{x} \). Evaluating \( \| A F_{\bar{p}}(\bar{x}) \|_{\infty} \) using this interval vector and interval arithmetic produces an interval that contains the true value of \( \| A F_{\bar{p}}(\bar{x}) \|_{\infty} \). We take its right endpoint as our rigorous bound \( Y \).

Observe that the bound $Y$ depends on the quality of the numerical approximation \(\bar{x}\). In practice, the bound can be improved by applying Newton's method using the validation map, and in principle it can be made as small as desired, up to the limits imposed by the employed floating-point precision.

The rigorous bound \( Z(r^{*}) \) for given $r^*$ is then obtained as the right endpoint of the rigorous evaluation of $\|I-ADF_{\bar{p}}(X)\|_\infty$, in which the entries of $X$ are intervals of radius $r^*$ centered at the entries of $\bar x$. Since the ball in \(\mathbb{R}^n\) with the infinity norm around \(\bar{p}\)
equals the Cartesian product of \(n\) intervals of radius \(r\) centered at the
components of \(\bar{p}\), we can represent it directly using interval
arithmetic.

Finally, we compute an upper bound on $r_{\min}$ by calculating \(Y/(1-Z(r^{*}))\) in interval arithmetic.
It remains to state how $r^*$ is chosen (the conditions in \Cref{theorem:validation} may hold for various values of \(r^{*}\)). In practice, we typically look for the biggest possible radius \(r^{*}\) such that approximately $Z(r^*)\leq\frac12$: This leads to a large uniqueness interval
\begin{equation}
\label{eq:radii_interval}
\left[  r_{\min}=\tfrac{Y}{1-Z(r^{*})}, r^{*} \right]
\end{equation}
in \Cref{theorem:minimizer} (increasing $r^*$ much further will not be possible anyway due to the condition $Z(r^*)<1$) and at the same time a high accuracy $r_{\min}$ of the approximated energy minimizer $\bar p$ (the accuracy could at most increase twofold by making $Z(r^*)$ arbitrarily small).
Such $r^*$ is found by a bisection method.

The positivity of all $(M_i(p)^{-1})_{ii}$ for $p\in B_{r_{\min}}$ is again checked by computing $(M_i(P)^{-1})_{ii}$ in interval arithmetic for $P_j=[p_j-r_{\min},p_j+r_{\min}]$ and validating whether the resulting interval is positive.

The code implementation and the data associated with the computer-assisted proofs in this paper are included in \cite{github_codes}. All implementations use the Julia programming language. Interval arithmetic is handled with the Julia libraries \textit{IntervalArithmetic.jl} \cite{IntervalArithmetic.jl} and \textit{Arblib.jl} \cite{Johansson2017Arb}.
For better readability, all numbers resulting from the rigorous numerics will be given with only 5 decimal places. We are now ready to explore various applications of \Cref{theorem:validation}.

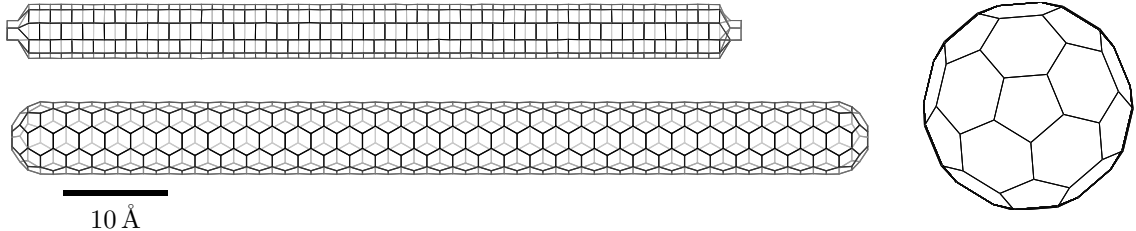
\begin{figure}[ht]
    \setlength\unitlength{\textwidth}
    \includegraphics[width=.70\unitlength,trim=575 445 450 425,clip]{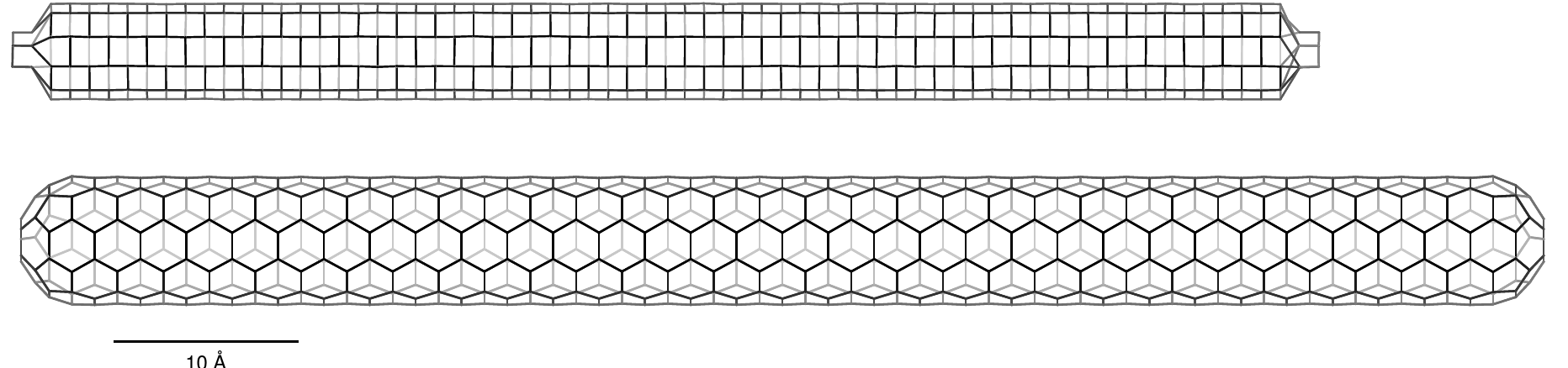}
    \includegraphics[width=.25\unitlength,trim=80 40 80 40,clip]{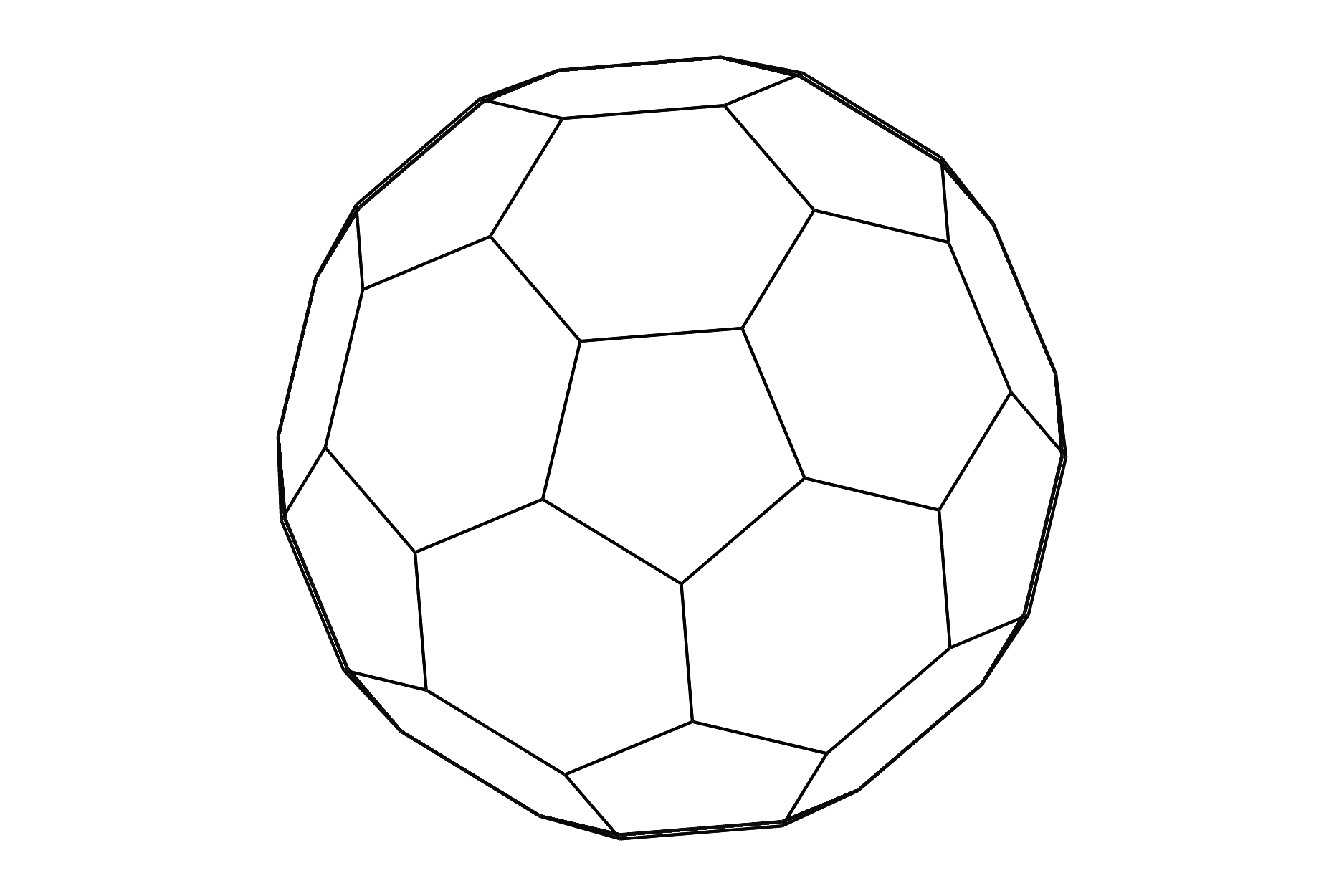}
    \begin{picture}(0,0)(1,0)
    \put(.054,.012){\rule{.0885\unitlength}{2pt}}
    \put(.07,-.015){\colorbox{white}{\small10\,\AA}}
    \end{picture}

    \caption{Top: Initialization of the BFGS gradient descent to find a local minimizer for a \((5,5)\)-armchair nanotube using the harmonic potential $\Eharmonic$.
    Bottom and right: Validated numerical approximation of the local minimizer (side and top view).
}
    \label{fig:nanotube_5_5_combined}
\end{figure}

\section{Carbon Nanotube Geometries}
\label{section:geometries}
As our first class of molecular structures, we apply our method to the study of the geometry of carbon nanotubes. The mathematical idealization of a carbon nanotube comes from rolling graphene. Graphene is a two-dimensional arrangement of carbon atoms where each atom forms three bonds with its neighbors. These bonds create a hexagonal lattice with a bond length of approximately $l=1.42$ \AA{} and an angle of \(2\pi/3\) between them. The graphene sheet can be fully described by two lattice vectors, \(a_1\) and \(a_2\), both with a length of \( a = l \sqrt{3} \)
(see \Cref{fig:graphene}). By rolling a graphene sheet into a cylinder, we obtain a carbon nanotube \cite{DRESSELHAUS1995883,meyyappan2004carbon}. The geometric classification of carbon nanotubes is based on the periodicity of this rolling process, which is given by a vector of the form
\[
    \mathbf{C} = na_1 + ma_2
\]
with integers \( n, m \). The resulting nanotube is said to have chirality \((n, m)\).

Two chiralities stand out:
Armchair nanotubes \((n, n)\) have a perfectly symmetric structure along their circumference, where the carbon-carbon bonds form a continuous chain resembling the armrests of a chair. Zigzag nanotubes (\(n, 0\)) feature carbon atoms aligned in straight rows along the tube's circumference, forming a repeating zigzag pattern along the edge. All other nanotubes \((n, m)\) with \(0 < |m| < n\), called chiral, exhibit a helical arrangement of atoms, breaking mirror symmetry and resulting in a more complex structure.

\begin{figure}[!ht]
  \centering

  \begin{minipage}[b]{0.5\textwidth}
    \centering
    \begin{overpic}[width=\linewidth,trim=60 120 100 100,clip]{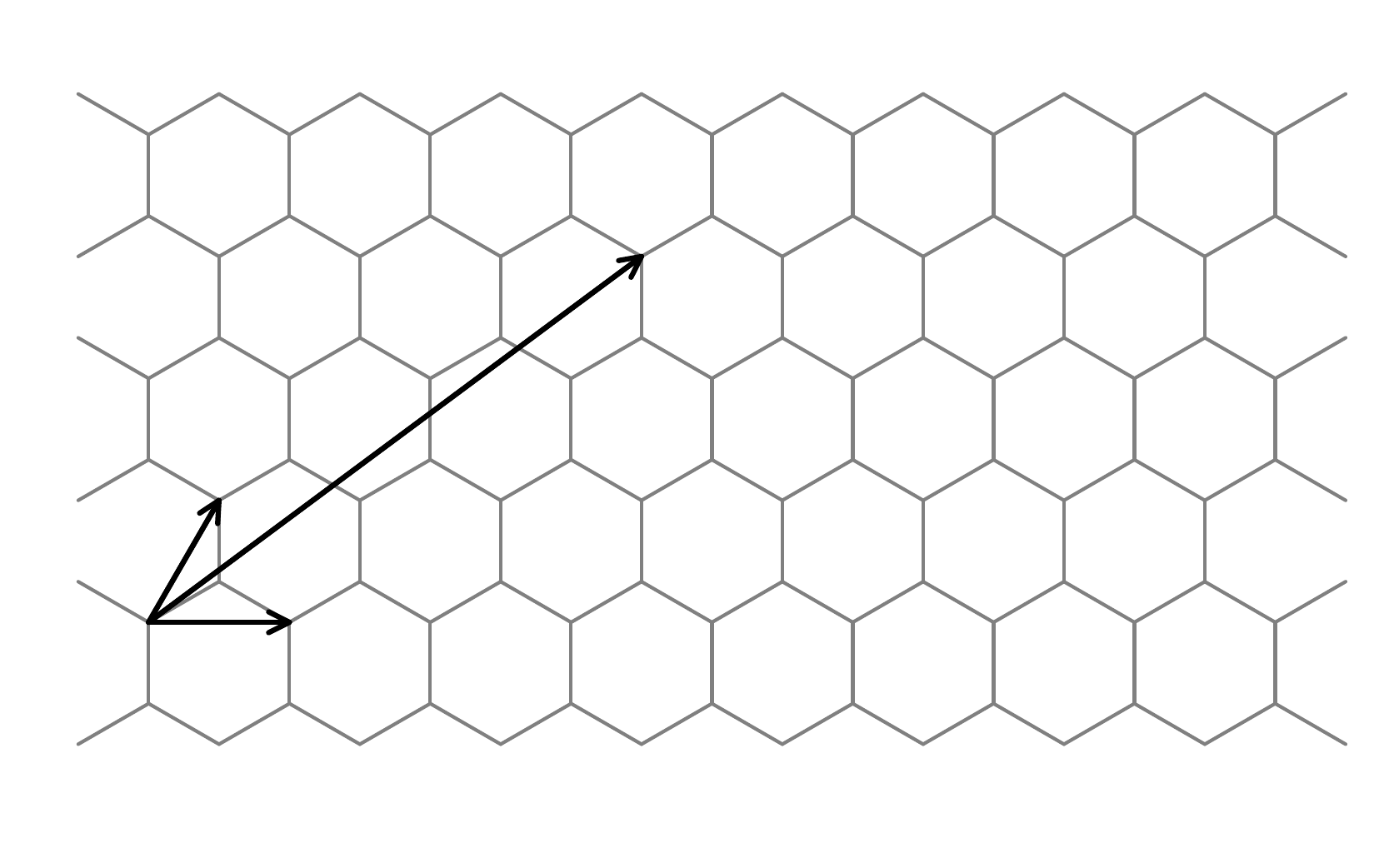}
      \put(50,32){\footnotesize $\mathbf{C} = n {a}_1 + m {a}_2$}
      \put(18,3){\footnotesize ${a}_1$}
      \put(3,13){\footnotesize ${a}_2$}
    \end{overpic}
  \end{minipage}
  \begin{minipage}[b]{0.35\textwidth}
    \includegraphics[width=\linewidth,clip]{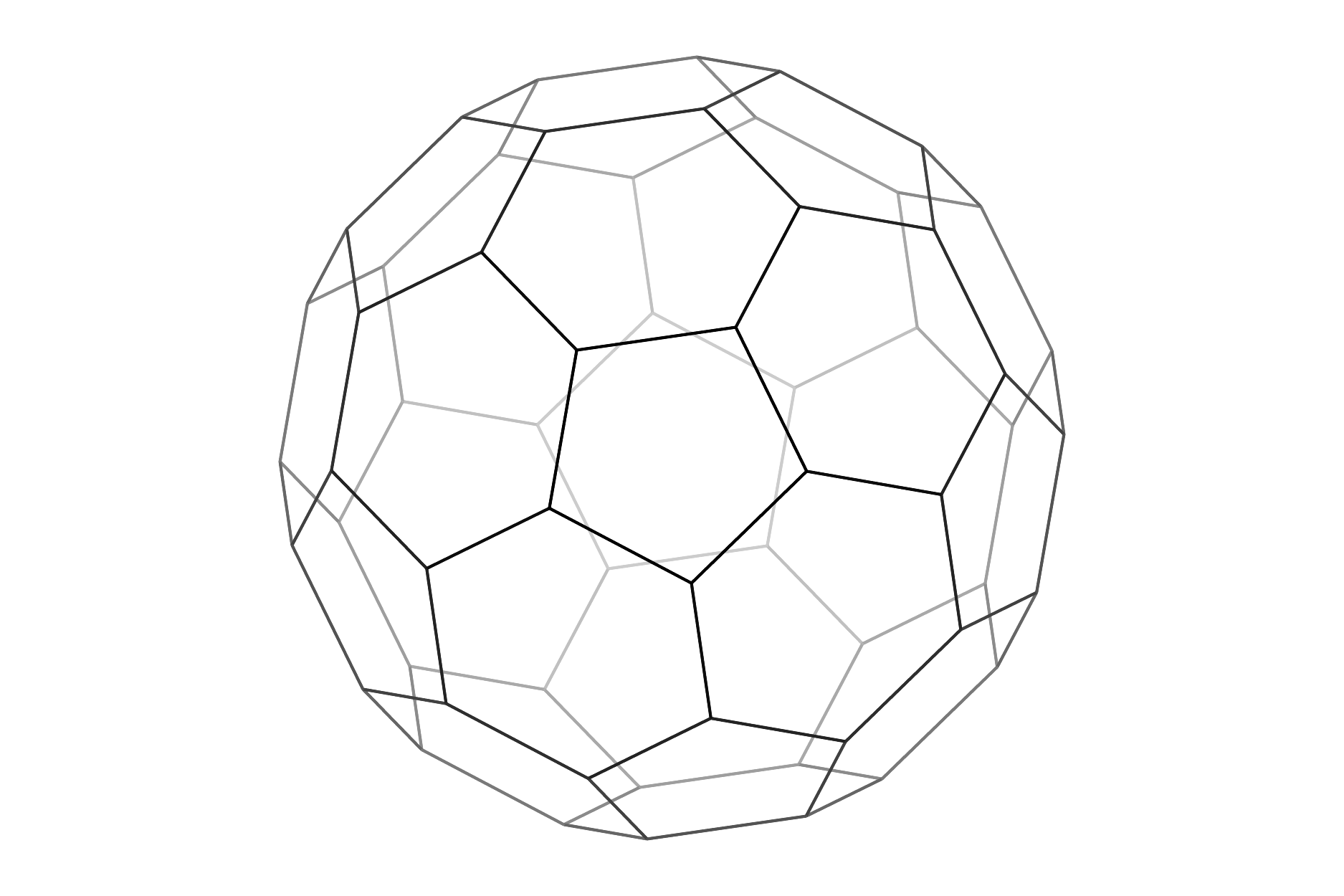}
  \end{minipage}

  \caption{Left: A portion of a graphene sheet with the chiral vector \(\mathbf{C} = n {a}_1 + m {a}_2\) indicated. Rolling the sheet along \(\mathbf{C}\) forms a carbon nanotube with chirality \((n, m)\). Right: Structure of a \(\mathrm{C}_{60}\) fullerene.}
  \label{fig:graphene}
\end{figure}

To describe a nanotube mathematically, we model each atom as a point atom in three-dimensional space. Interactions between atoms are represented by an \emph{interatomic potential}, which gives us a model for the total energy of the system as a function of the atomic positions. This potential can be obtained from quantum mechanical calculations or from experimental data. In practice, analytic approximations of this potential are used to model the interatomic forces.

The geometry of a carbon structure is then obtained by minimizing an analytic total potential energy with respect to the atomic positions. This procedure is computational, since it requires handling a nonlinear energy function in a high-dimensional configuration space. The study of simulations produced by such minimization fits directly into the rigorous geometry optimization framework presented in \Cref{section:CAPs}. We now examine a scenario in which the potential energy function remains fixed while we modify the geometry of the structures. In this context, we investigate the structural differences that can be rigorously observed when varying the chirality or cap configuration while employing the same interatomic potential.

For now, we focus on interactions that arise from atomic bonding. This is particularly relevant for materials like carbon nanotubes and graphene, where bond lengths and angles play an important role in determining structural and electronic properties. To this end we employ the energy $\Eharmonic$ from \eqref{eq:energy_harmonic} with harmonic two-atom and three-atom interaction potentials,
\begin{equation}
    \label{eq:energy_harmonic}
    \Eharmonic(p) \bydef  \sum_{(i,j) \in B} E_{\bond}(r_{ij})  + \sum_{(i,j,k) \in A} E_{\angle}(\theta_{ijk}),
    \quad\text{with }
    E_{\bond}(r) = k_{b} (r - r_0)^2,\;
    E_{\angle}(\theta) = k_{\theta} (\theta - \theta_0)^2,
\end{equation}
For the modeling of carbon atoms, we use the parameters \( k_{b} = 469 \) kcal/mol/\AA\(^2\), \( r_0 = 1.44 \) \AA, \( k_{\theta} = 63 \), and \( \theta_0 = 2\pi/3 \), taken from \cite{SaitoBook, Weiner1986}. The set \(B\) is defined as follows. For a given carbon nanotube geometry, for example the one illustrated in
Figure~\ref{fig:nanotube_5_5_combined}, we index the atoms in the tube by
\(
    V = \{1,\dots,N\}
\)
and define the \emph{neighbor list} by
\[
B \;=\; \bigl\{\,(i,j)\subset V \;:\; i\neq j \text{ and atoms } i \text{ and } j \text{ are first neighbors (separated by one bond) }\,\bigr\}
\]
and the \emph{angle set} as
\[
A \;=\; \bigl\{\, (i,j,k)\in V^3 \;:\; i,j,k \text{ distinct and } ij, jk \in B \,\bigr\}.
\]
We define the bond length and bond angle as
\[
  r_{ij} = \|p_i - p_j\| \quad \text{for } (i,j) \in B,
  \qquad
  \theta_{ijk}
  = \arccos\!\left(
    \frac{(p_i - p_j)\cdot(p_k - p_j)}{\|p_i - p_j\|\,\|p_k - p_j\|}
  \right)
  \quad \text{for } (i,j,k) \in A.
\]

The above $\Eharmonic$ may be viewed as the second-order approximation of a potential energy that includes interactions between neighboring atoms and neighboring atom pairs. Furthermore, the harmonic potential \eqref{eq:energy_harmonic} can be seen as a second-order approximation of the Morse potential \cite{morse1929diatomic} near its minimum. Thus it is the simplest possible approximation to more complicated potentials, and the natural question arises whether this approximation exhibits any qualitatively different behavior. Let us recall that our approach applies to any interaction potential and the above choice represents just an example. More complicated potential energies exist and are used in the literature (such as Tersoff's potential to be introduced in \eqref{eq:energy_tersoff}).

For our first example, we analyze a (5,5)-armchair nanotube with 670 atoms. This chirality allows only one fullerene (see \Cref{fig:graphene}) cap configuration, with pentagons at the center, keeping the 5-fold symmetry. Each cap consists of 30 atoms, while each cross section contains 10 atoms, resulting in a total of 31 cross sections. We sometimes refer to cross sections as \textit{rings}. Once we select the geometry to analyze, the first step is to approximate the solution.

As initialization, for the study of capped nanotubes we simply place atoms in a regular pattern on a cylindrical tube, where the exact pattern depends on the chosen nanotube chirality (see \Cref{fig:nanotube_5_5_combined} for an example). \Cref{fig:nanotube_5_5_combined} shows a numerical approximation of a local minimizer of the energy function \eqref{eq:energy_harmonic} for a (5,5)-armchair nanotube with caps, which lies within a $5\cdot10^{-11}$ \AA{} distance from the true local minimizer.

\begin{theorem}[Capped armchair nanotube]
    \label{theorem:existence_armchair_harmonic}
    Consider the approximate local minimizer $\bar{p}\in \mathbb{R}^{3n}$ of
    $\Eharmonic$ from \eqref{eq:energy_harmonic}, with bonds corresponding to
    the $(5,5)$-armchair capped nanotube with $n=670$ shown in
    \Cref{fig:nanotube_5_5_combined}. There exists a local minimizer $p$ of
    $\Eharmonic$ such that
    \[
        \|p - \bar{p}\|_\infty \le 4.6844\cdot 10^{-11}\,\text{\AA}.
    \]
\end{theorem}
\begin{proof}
This is an application of \Cref{theorem:minimizer} with numerically found $\bar p$ as well as $\bar\tau,\bar\rho,A$
and all conditions rigorously verified numerically via interval arithmetic
for the values $r^*=9.9993\cdot10^{-8}$ \AA{}, $Z(r^*)=5.3880\cdot10^{-1}$ \AA{}, and $Y=2.1601\cdot10^{-11}$ \AA{}, resulting in $r_{\min}=Y/(1-Z(r^*))\leq4.6844\cdot10^{-11}$ \AA{}. We obtain a true local minimizer of $\Eharmonic$ that lies within distance $4.6844\cdot10^{-11}$ \AA{} of $\bar p$ and is unique within distance $9.9993\cdot10^{-8}$ \AA{} of $\bar p$ (in the norm $\|\cdot\|_\infty$). 

The detailed computational procedure was already described after \Cref{theorem:minimizer}.
\Cref{fig:ZDependence} illustrates the numerical dependence of $Z(r^*)$ and $r_{\min}$ on the choice of $r^*$ (recall that we choose $r^*$ so that $Z(r^*)\approx1/2$).
It shows that, employing different choices of $r^*$,
one can even improve the accuracy estimate $r_{\min}$ for the approximate minimizer by a factor $2$
and increase the uniqueness radius $r^*$ by a factor of $6$.
\end{proof}

\begin{figure}
\centering
\begin{tikzpicture}
\begin{loglogaxis}[
height=.3\textwidth,
width=.48\textwidth,
xlabel={$r^*$ },
ylabel={$Z(r^*)$},
]
\addplot[black, dotted, mark=square*] coordinates {(1e-10,0.00053889) (1e-09,0.0053889) (1e-08,0.053889) (1e-07,0.53889) (5.9605e-07,0.97491) (6.1138e-07,0.99999)};
\end{loglogaxis}
\end{tikzpicture}
\begin{tikzpicture}
\begin{loglogaxis}[
height=.3\textwidth,
width=.48\textwidth,
xlabel={$r^*$ },
ylabel={$r_{\min}$ },
]
\addplot[black, dotted, mark=square*] coordinates {(1e-10,2.1613e-11) (1e-09,2.1719e-11) (1e-08,2.2832e-11) (1e-07,4.6847e-11) (5.9605e-07,1.6238e-10) (6.1138e-07,5.6036e-07)};
\end{loglogaxis}
\end{tikzpicture}
\caption{Dependence of $Z(r^*)$ and $r_{\min}$ on $r^*$ for the example of \Cref{theorem:existence_armchair_harmonic}.}
\label{fig:ZDependence}
\end{figure}
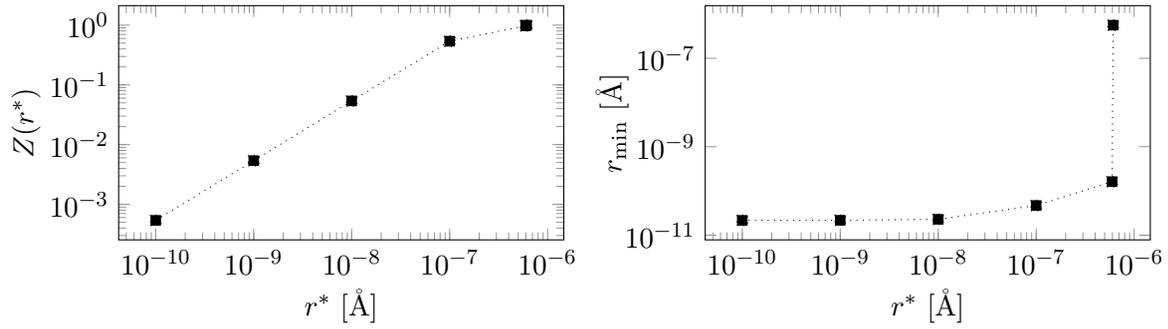

The existence radius $4.6844\cdot 10^{-11}$ in the statement of
Theorem~\ref{theorem:existence_armchair_harmonic} should not be read as a
physical length scale. It is a rigorous bound obtained from our validation
theorem and shows that the approximate configuration differs from a true local
minimizer by at most $4.6844\cdot 10^{-11}$ in the $\|\cdot\|_\infty$-norm.
The fact that this radius is so small means that the numerical approximation is
very close to a true local minimizer.

More concretely, a bound of the form
\(
  \|p - \bar{p}\|_\infty \le r
\)
means that each atomic position of $p$ lies in a small cube of side
length $2r$ centered at its numerical value, and this interval
representation of the local minimizer can be implemented rigorously using
interval arithmetic.

In the same spirit, the uniqueness radius appearing in the proof of
Theorem~\ref{theorem:existence_armchair_harmonic} is also a mathematical
quantity rather than a physical scale. It shows that no other local minimizer
lies closer to the numerical approximation than this radius. This does not rule
out the possibility of other, qualitatively similar configurations further away
in the energy landscape.

These bounds on the atomic positions allow us to control any geometric quantity
computed from the configuration, such as bond lengths, bond angles, or tube
radius. Each of these quantities comes with a mathematically proven error bar
derived from the position bounds.

The relevance of these rigorous statements depends on the questions we want to
ask. In many physical applications one is interested in differences that appear
in the first three or four significant digits. Smaller discrepancies may 
not be relevant. From the point of view of a mathematical study, however, there
is no restriction on how small a rigorously detected difference can be. In
fact, working with nine or ten significant digits (or even more) is standard
practice.

What can we learn from \Cref{theorem:existence_armchair_harmonic}?
We begin by examining the impact of the caps on the nanotube's radius. To this end we first align the computed nanotube's axis with the \( z \)-axis and compute the distance from each atom to the \( z \)-axis. Together with \Cref{theorem:existence_armchair_harmonic} this produces an interval for each atom, which we represent as error bars in our plots. The intervals are small enough that any visible distinctions in the plot are
mathematically rigorous. In other words, what we see is what we can prove.

We present our results in \Cref{fig:nanotube_5_5_relative radii}. The horizontal axis indicates the index of each cross section, with zero corresponding to the middle of the tube. The vertical axis, shown on a logarithmic scale, displays the difference between each atom's distance to the nanotube axis and the average distance in the middle cross section, excluding atoms in the caps. Each visible dot corresponds to ten overlapping error bars (one for each atom of the cross section), showing the interval within which the deviation lies for the true energy minimizer.

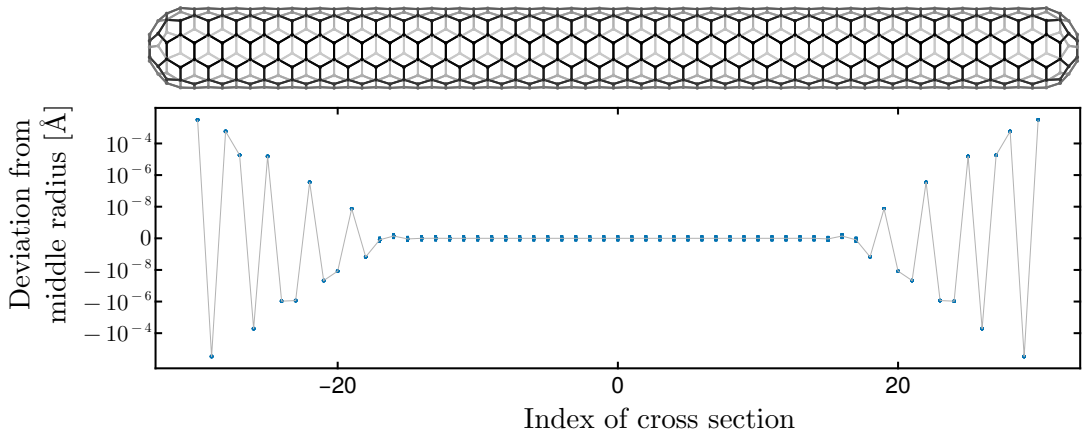
\begin{figure}[!ht]
    \hspace*{2.5em}\begin{minipage}{\linewidth}
 \includegraphics[width=0.9\textwidth,trim=150 250 200 190,clip]{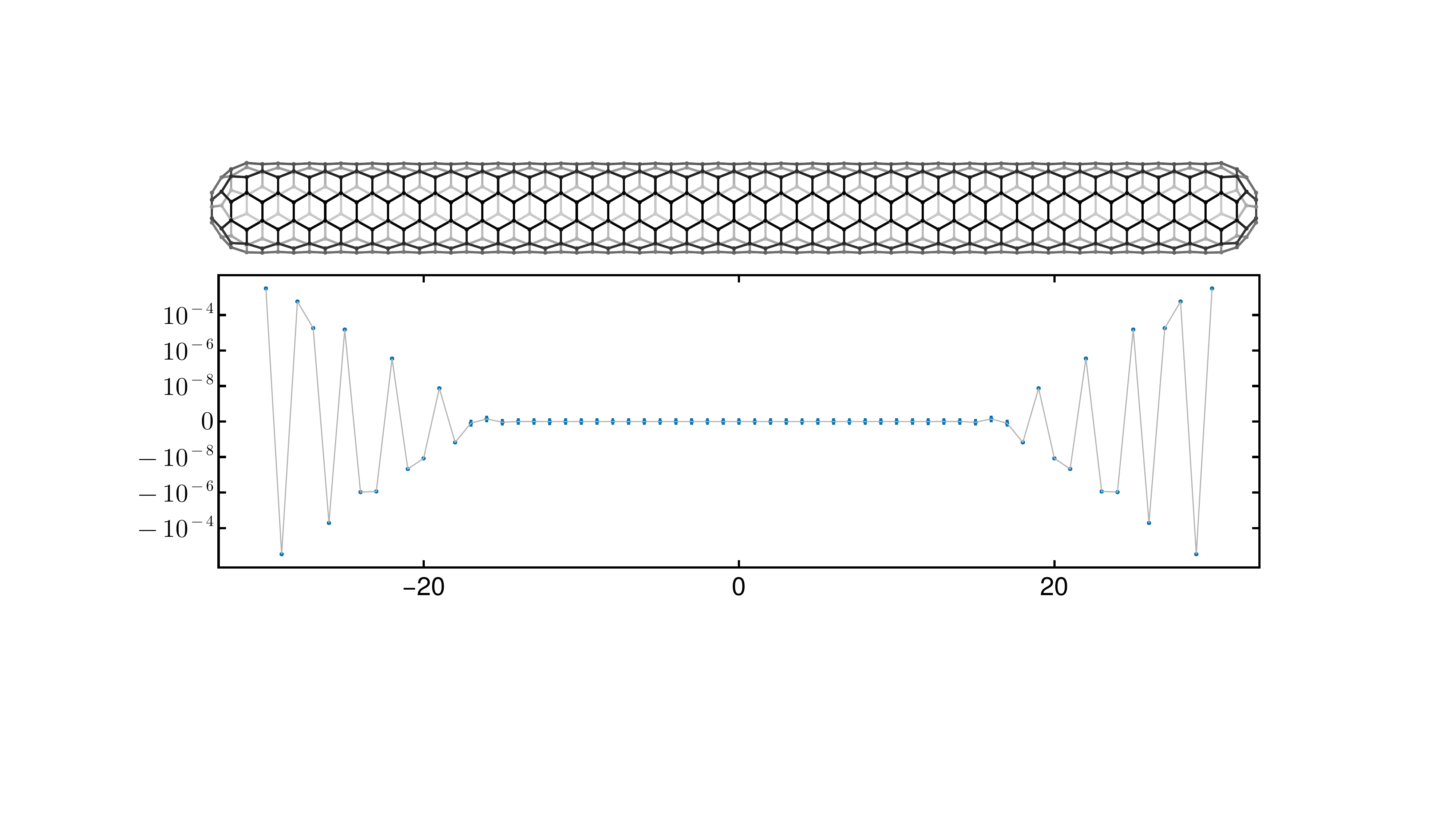}

\setlength\unitlength{\textwidth}
\begin{picture}(1,0)(0,0)

  \put(0.5, 0.055){
    \makebox(0,0)[t]{\colorbox{white}{\strut Index of cross section}}
  }

  \put(-0.060, 0.09){
    \rotatebox{90}{
      \colorbox{white}{
        \begin{minipage}{0.19\textwidth}
          \centering
          Deviation from\\middle radius [\AA{}]
        \end{minipage}
      }
    }
  }

\end{picture}
\end{minipage}

  \caption{A \((5,5)\)-armchair nanotube configuration with 670 atoms for the harmonic potential $\Eharmonic$ proven to be energy minimizing with computer assistance. The plot underneath is aligned with the nanotube and shows the deviation of the distance between each atom and the nanotube axis from the mean distance in cross section 0. Each dot represents ten overlapping error bars (too small to be resolved in the figure), one for each atom of the cross section.
}
     \label{fig:nanotube_5_5_relative radii}
\end{figure}

For example, the slight misfit between the native cap and nanotube diameter provably leads to oscillations in the nanotube diameter, shown in \Cref{fig:nanotube_5_5_relative radii}.
The amplitude of these oscillations decreases exponentially away from the caps until at cross section $\pm17$ it reaches a size that can no longer be resolved by the numerical accuracy of our proof
(but would be resolved when working with higher numerical precision).
\Cref{fig:middle_section_5_5_harmonic} zooms in on cross sections $-17$ to $17$ so that the error bars become visible.
As can be seen, the diameters of these cross sections (distance between opposite atoms) coincide up to $10$ digits of accuracy, taking the value of $6.949375982$ \AA.
As expected from the exponential decay of the cap influence, also a shorter nanotube shows the exact same behavior, see \Cref{fig:nanotube_5_5_short_relative radii}.
This remains true for a nanotube with an even number of cross sections as shown in \Cref{fig:nanotube_5_5_N360} (only tube configuration shown), where the reflection symmetry across the nanotube's midplane is lost
(the central pentagon of one cap is rotated relative to the other).

\begin{figure}[!ht]
    \centering
    \includegraphics[width=\textwidth,trim=0 0 0 0,clip]{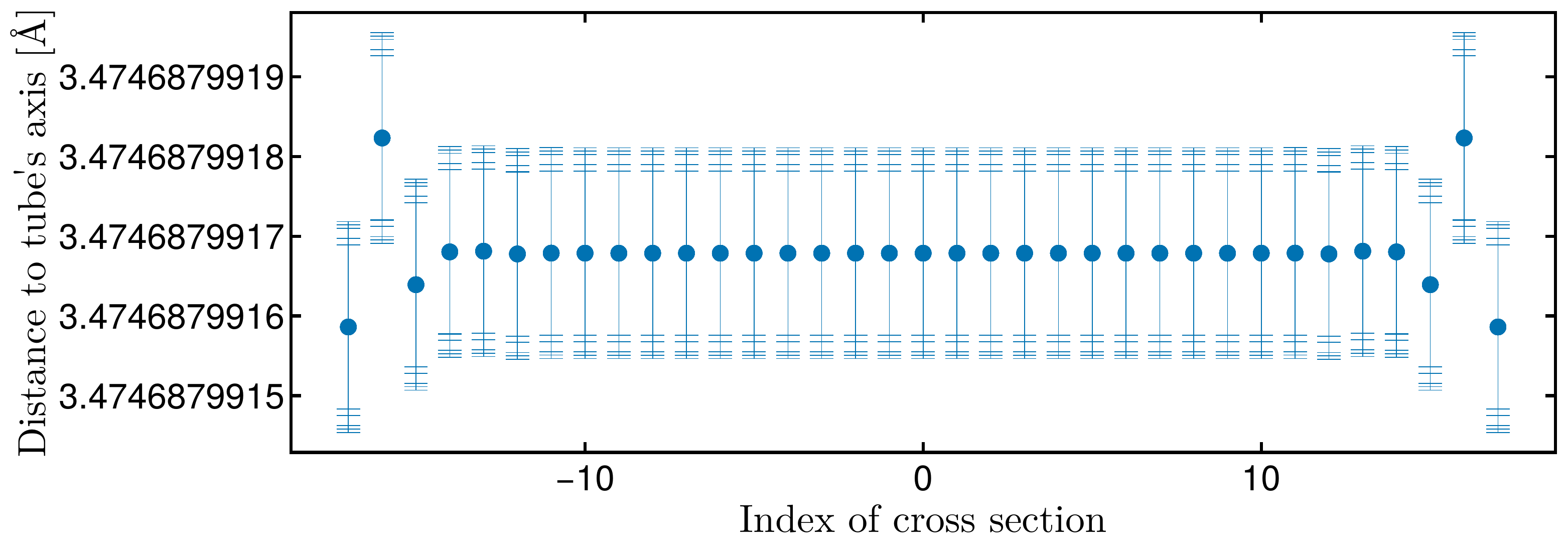}

    \caption{A zoom-in of the middle section of the plot in \Cref{fig:nanotube_5_5_relative radii}. The ten overlapping error intervals are now visible. The dots represent their midpoints. In this section of the tube, the computer-assisted proof, with the current accuracy, does not allow us to detect oscillations in the radius.}
    \label{fig:middle_section_5_5_harmonic}
\end{figure}

\begin{figure}[!ht]
    \centering
    \includegraphics[width=\textwidth,trim=0 255 0 100,clip]{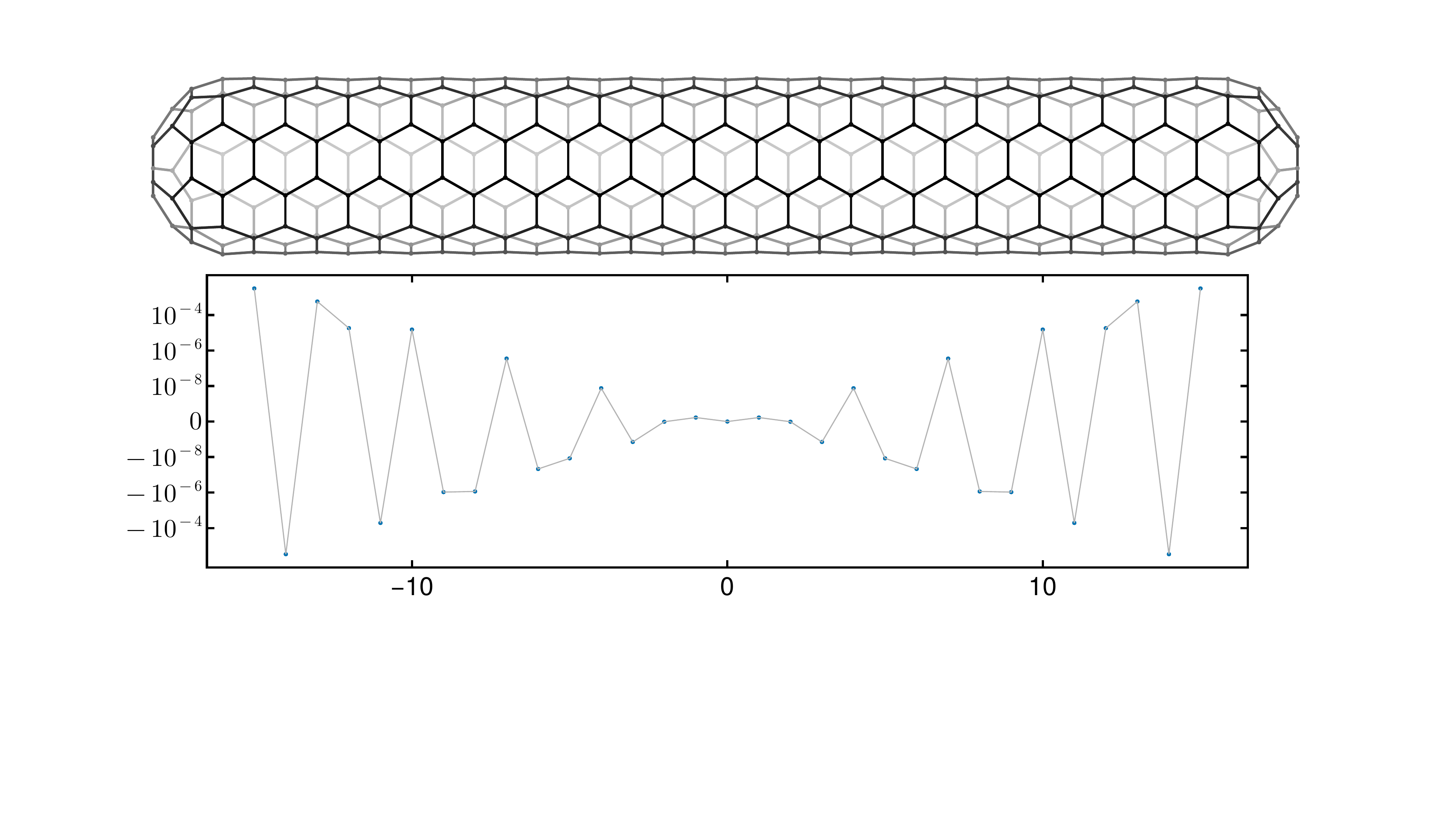}
\setlength\unitlength{\textwidth}
\begin{picture}(1,0)(0,0)

  \put(0.48, 0.04){
    \makebox(0,0)[t]{\colorbox{white}{\footnotesize  \strut Index of cross section}}
  }

  \put(0.04, 0.07){
    \rotatebox{90}{
      \colorbox{white}{
        \begin{minipage}{0.19\textwidth}
          \centering \footnotesize
          Deviation from\\middle radius [\AA{}]
        \end{minipage}
      }
    }
  }

\end{picture}

    \caption{Same as \Cref{fig:nanotube_5_5_relative radii} for an energy-minimizing \((5,5)\)-armchair nanotube with 370 atoms.
    }
    \label{fig:nanotube_5_5_short_relative radii}
\end{figure}
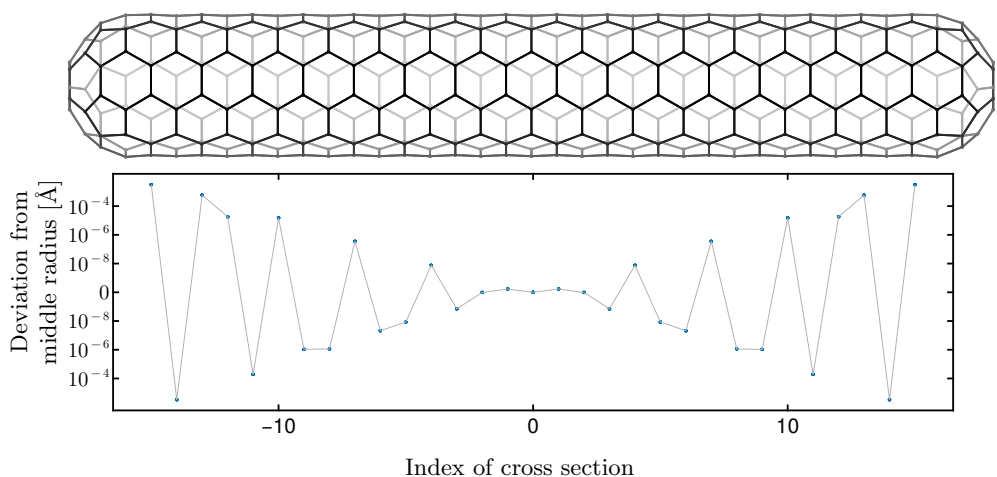

\begin{figure}[ht]
    \centering

    \begin{subfigure}[b]{0.65\textwidth}
        \centering
        \includegraphics[width=\textwidth,trim=40 200 30 205,clip]{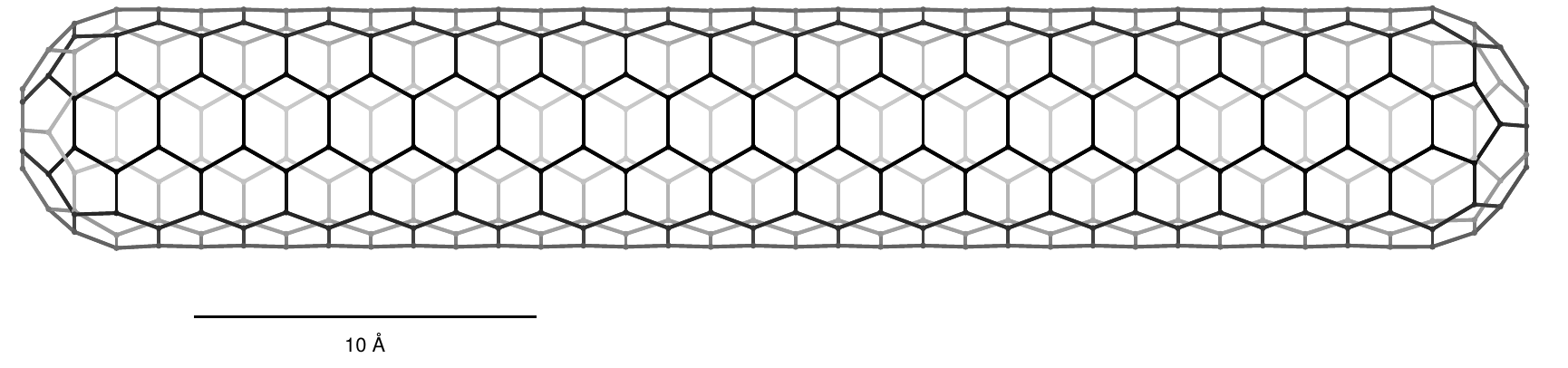}

          \setlength{\unitlength}{1mm}
          \begin{picture}(0,0)
            \put(-38.35,7.5){\rule{22.5mm}{2pt}}
            \put(-30.5,3.4){\colorbox{white}{\small10\,\AA}}
          \end{picture}     
    \end{subfigure}
    \hfill
    \begin{subfigure}[b]{0.3\textwidth}
        \centering
        \includegraphics[width=\textwidth]{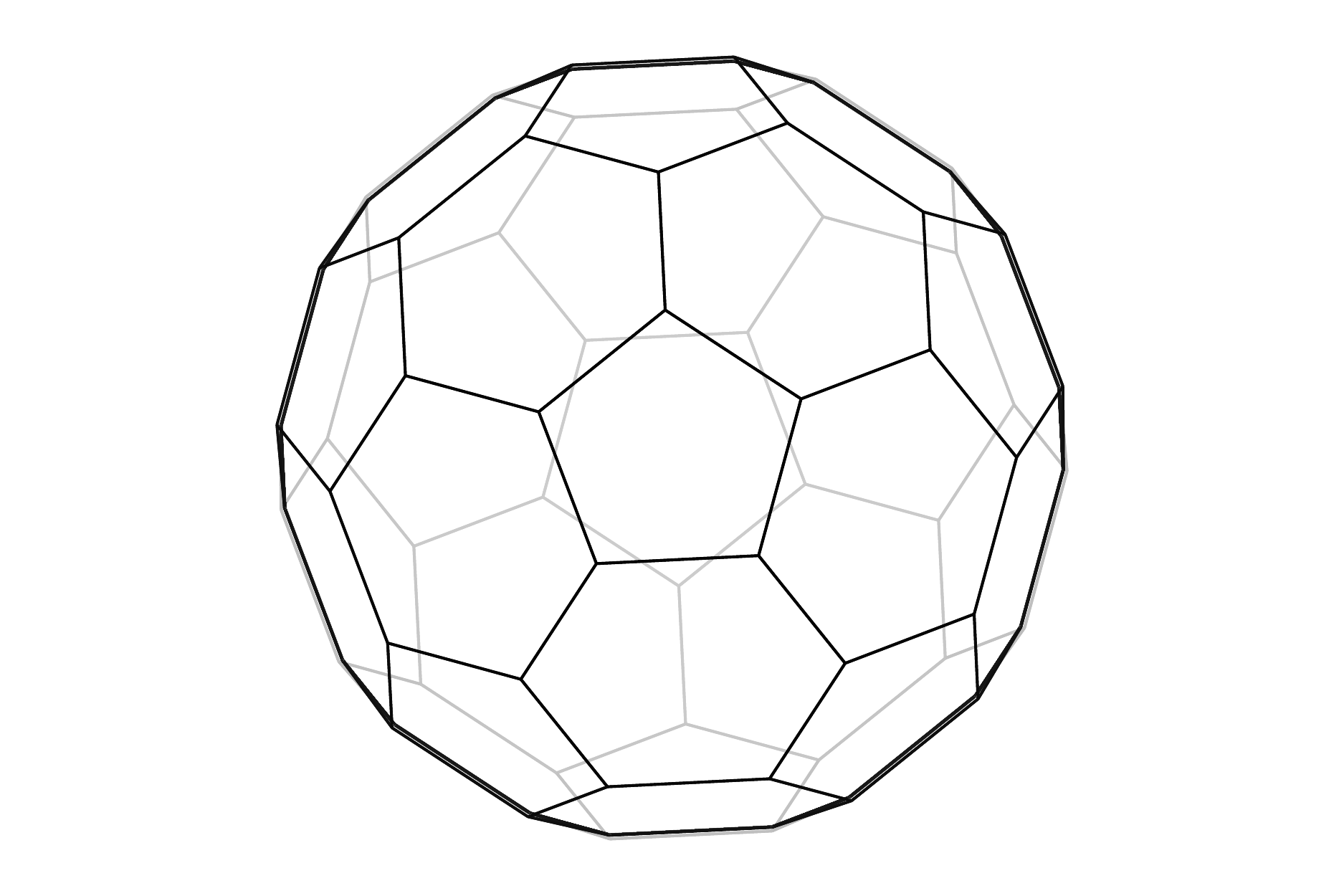}
    \end{subfigure}

\caption{A numerically validated approximation of a local energy minimizer for a capped \((5,5)\)-armchair nanotube with an even number of cross sections, using the harmonic potential $\Eharmonic$. In this case, the reflection symmetry is lost (compare \Cref{fig:nanotube_5_5_combined}).}

    \label{fig:nanotube_5_5_N360}
\end{figure}

\begin{remark}[The radius of the nanotube]
Carbon nanotube structures are not perfect cylinders, but polyhedral. When we refer to the \textit{radius of the tube}, we mean the distance from an atom to the axis of the tube.
The radius of an ideal $(n,m)$ nanotube cylinder of bond length $r_0$ is given by \cite{BUDYKA2005266}
\[
    \frac{r_0}{2\pi}  \sqrt{3(n^2 + m^2 +nm)}.
\]
Since bond lengths in carbon nanotubes slightly deviate from those in planar graphene, we adopt a bond length of $r_0=1.44$ \AA{} \cite{SaitoBook}. For the (5,5) armchair nanotube this results in a radius of $3.4377$ \AA{}. 

This geometric formula provides only an idealized reference point. In practice, the usual formula as a function of $(n,m)$ is reliable only when the tube radius is sufficiently large. For small radii, curvature and electronic effects become important, and some interatomic potentials can miss these contributions. 
One possible use of computer-assisted validation is to quantify when such continuum or rolled-graphene idealizations become accurate enough for a given interatomic potential.
\end{remark}

Let us now focus our analysis on a different geometry: the (10,0)-zigzag nanotube with caps consisting of 660 atoms. This chirality features cross sections composed of ten atoms, similar to the (5,5)-armchair nanotube previously presented. We first rigorously validate the numerical approximation shown in \Cref{fig:nanotube_10_0_N660}.
The proof is analogous to that of \Cref{theorem:existence_armchair_harmonic}.

\begin{theorem}[Capped zigzag nanotube]
    \label{theorem:existence_zigzag_harmonic}
    Consider the approximate local minimizer $\bar{p}\in \mathbb{R}^{3n}$ of $\Eharmonic$ from \eqref{eq:energy_harmonic} (with bonds corresponding to the (10,0)-zigzag capped nanotube with $n=660$) shown in \Cref{fig:nanotube_10_0_N660}. There exists a local minimizer $p$ of
    $\Eharmonic$ such that
    \[
        \|p - \bar{p}\|_\infty \le 1.2226\cdot10^{-11}\,\text{\AA}.
    \]
\end{theorem}

\begin{proof}
The proof is analogous to that of Theorem~\ref{theorem:existence_armchair_harmonic}
and follows the same validation procedure. In this case, we obtain a true local minimizer of $\Eharmonic$ that lies within distance $1.2226\cdot10^{-11}$ \AA{} of $\bar p$ and is unique within distance $9.9902\cdot10^{-8}$ \AA{} of $\bar p$ in the norm $\|\cdot\|_\infty$.
\end{proof}

\begin{figure}[ht]
    \centering

    \begin{subfigure}[b]{0.65\textwidth}
        \centering
        \includegraphics[width=\textwidth,trim=40 200 50 230,clip]{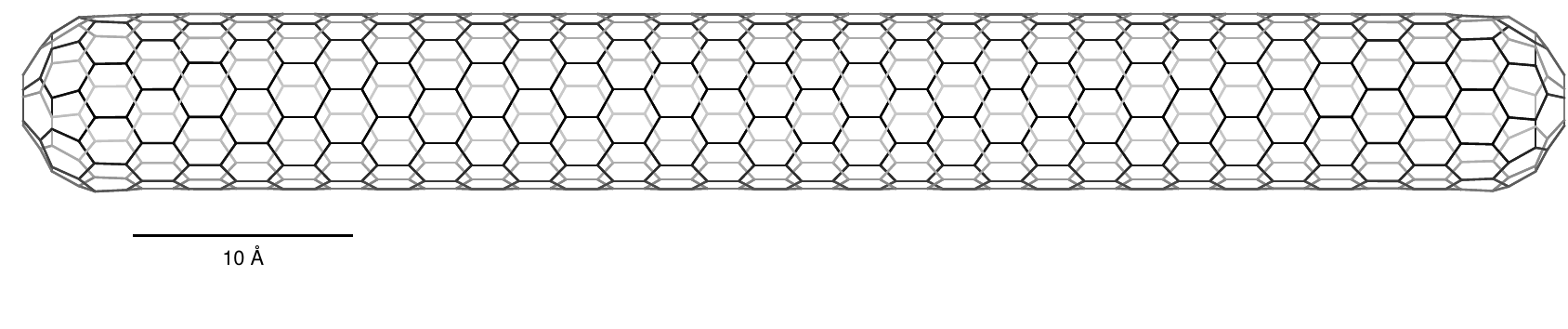}
    \end{subfigure}
    \hfill
    \begin{subfigure}[b]{0.3\textwidth}
        \centering
        \includegraphics[width=\textwidth,trim=0 40 0 40,clip]{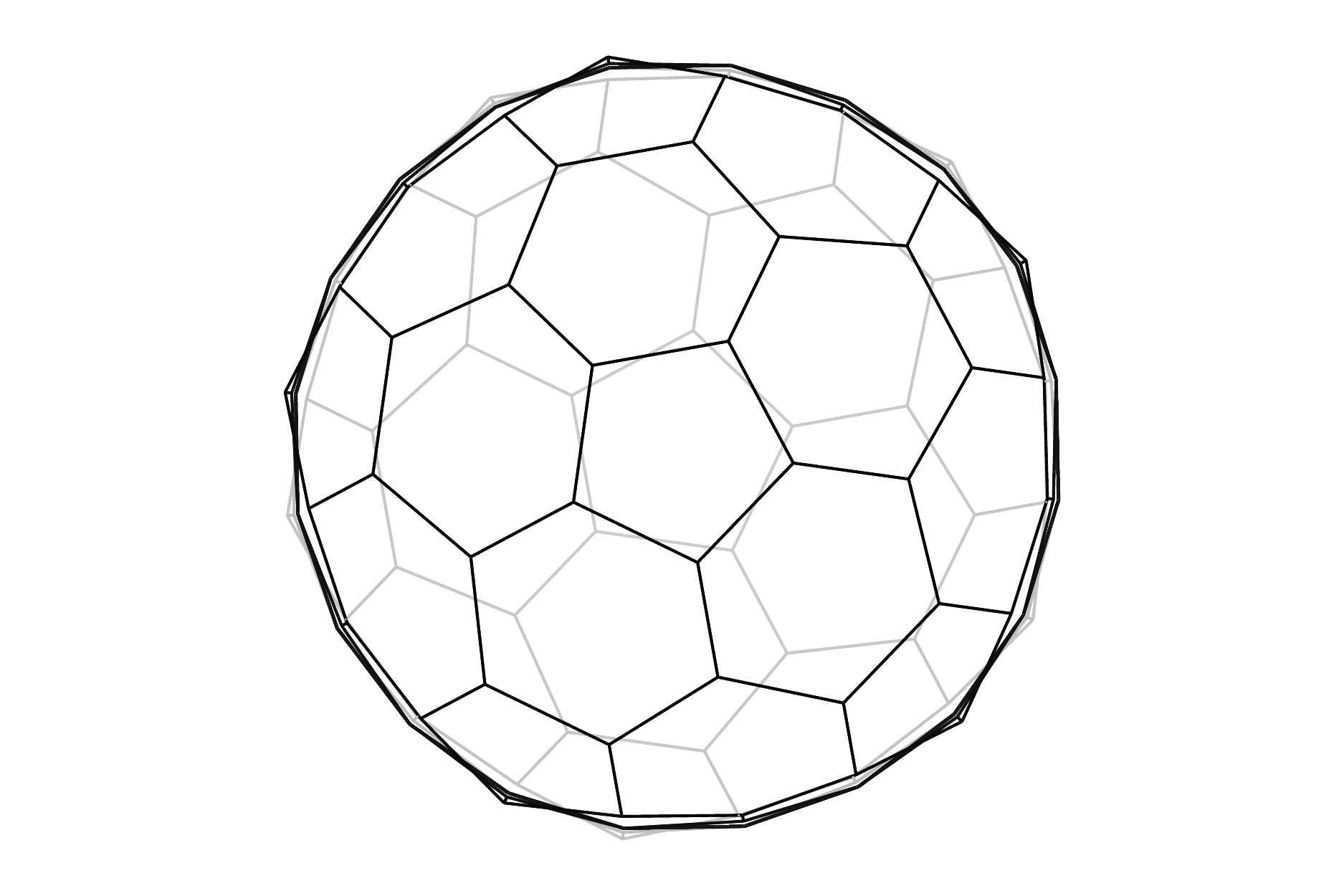}

          \setlength{\unitlength}{1mm}
          \begin{picture}(0,0)
            \put(-124.3,5.5){\rule{14.35mm}{2pt}}
            \put(-121,1.3){\colorbox{white}{\small10\,\AA}}
          \end{picture}
        
    \end{subfigure}

        \caption{A numerically validated approximation of a local energy minimizer for a \((10,0)\)-zigzag nanotube with pentagonally symmetric caps, using the harmonic potential $\Eharmonic$.
        }
    \label{fig:nanotube_10_0_N660}
\end{figure}
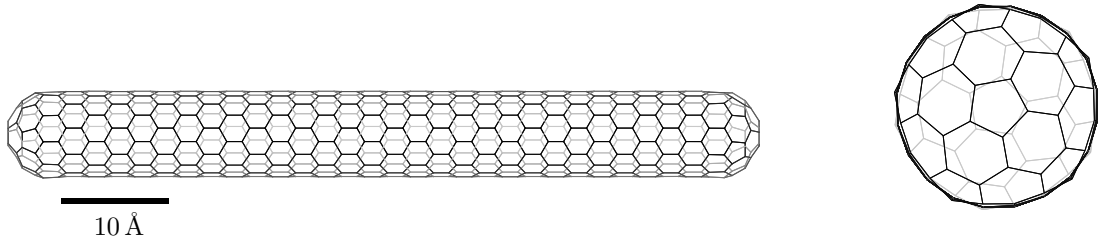

Just like for the $(5,5)$-armchair nanotube, the diameters of the cross sections oscillate with amplitude exponentially decreasing away from the caps, but the type of oscillations is qualitatively different as illustrated in \Cref{fig:nanotube_10_0_radii}:
Every third and fourth cross section has perfect diameter (the same as the middle cross sections),
while for every other cross section half the atoms have a larger and the other half a smaller distance from the nanotube axis
(where the deviation to larger and smaller distances is symmetric).

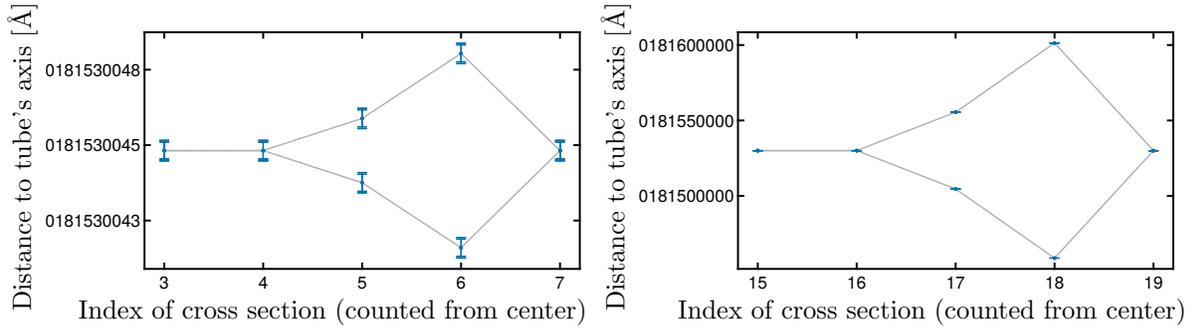
\begin{figure}[ht]
\centering
\setlength\unitlength\textwidth
\begin{picture}(1,.25)
\put(.025,0){\includegraphics[height=.25\unitlength,trim=45 0 0 0,clip]{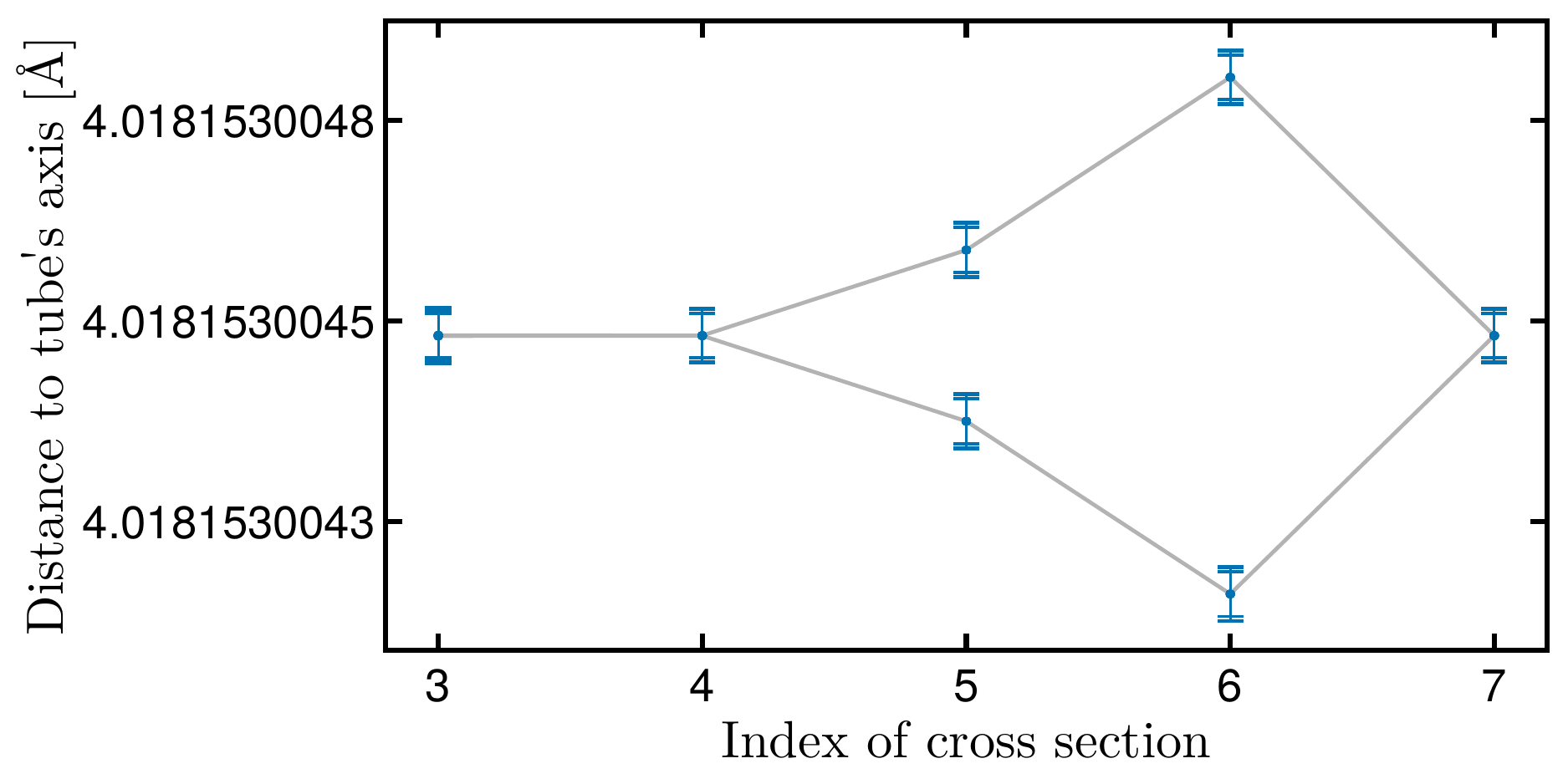}}
\put(.53,0){\includegraphics[height=.25\unitlength,trim=45 0 0 0,clip]{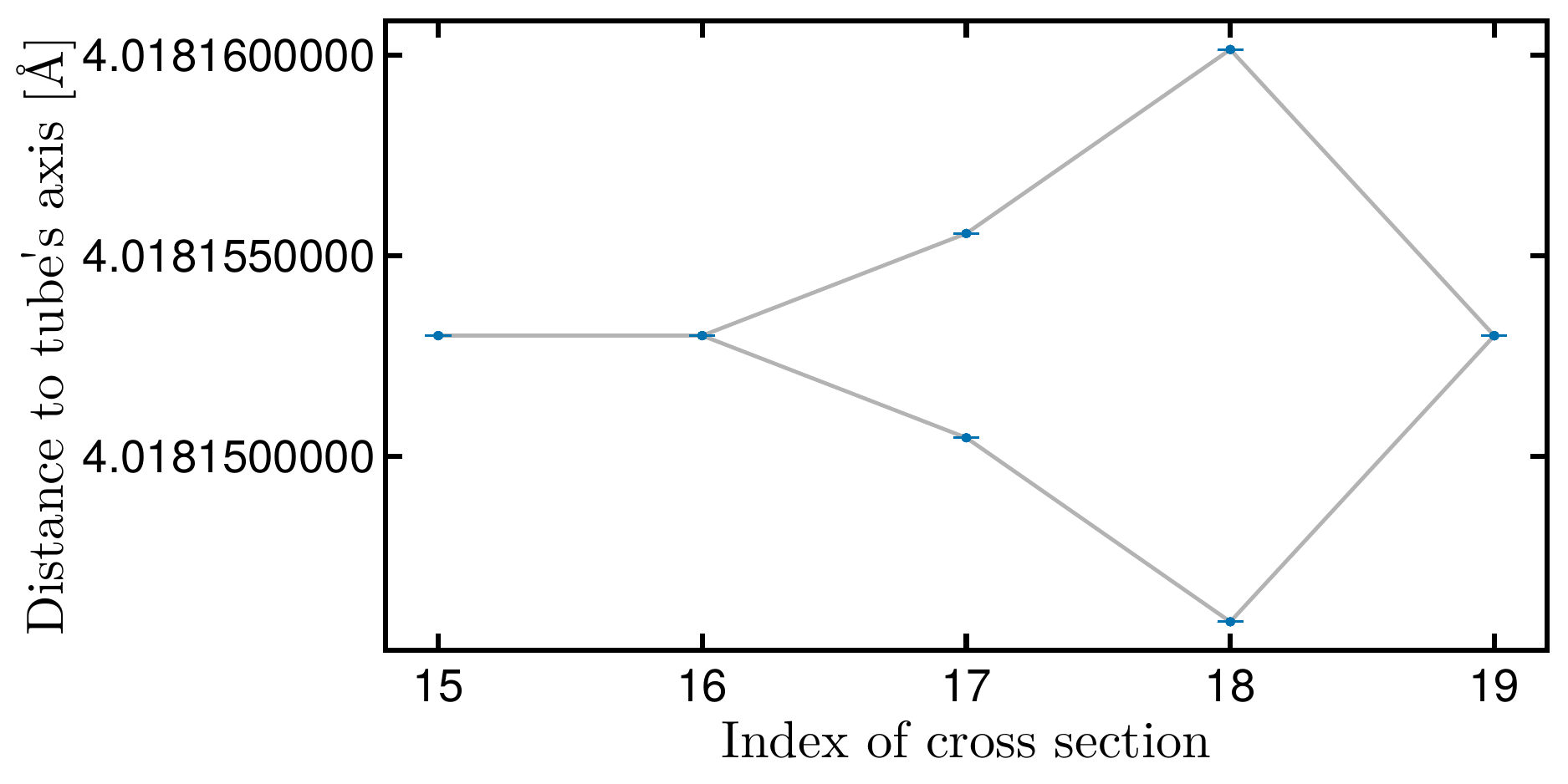}}
\put(.06,0){\colorbox{white}{\small Index of cross section (counted from center)}}
\put(.57,0){\colorbox{white}{\small Index of cross section (counted from center)}}
\put(-0.01,.01){\rotatebox{90}{\colorbox{white}{\footnotesize   Distance to tube's axis [\AA{}]}}}
\put(.496,.01){\rotatebox{90}{\colorbox{white}{\footnotesize Distance to tube's axis [\AA{}]}}}
\end{picture}

  \caption{Distance of each atom to the nanotube axis for the \((10,0)\)-zigzag validated nanotube from \Cref{theorem:existence_zigzag_harmonic}.}

    \label{fig:nanotube_10_0_radii}
\end{figure}

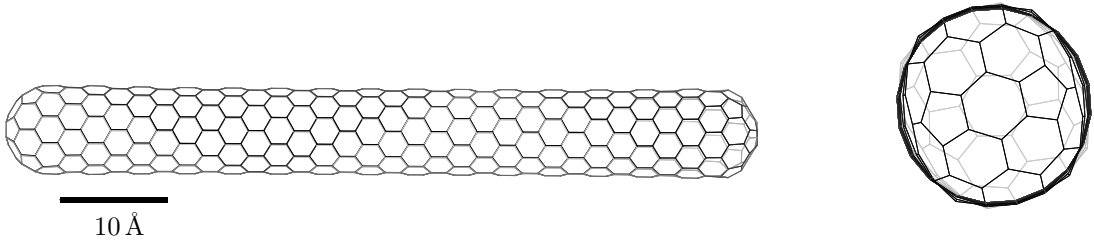
\begin{figure}[ht]
    \centering

    \begin{subfigure}[b]{0.65\textwidth}
        \centering
        \includegraphics[width=\textwidth,trim=40 220 50 205,clip]{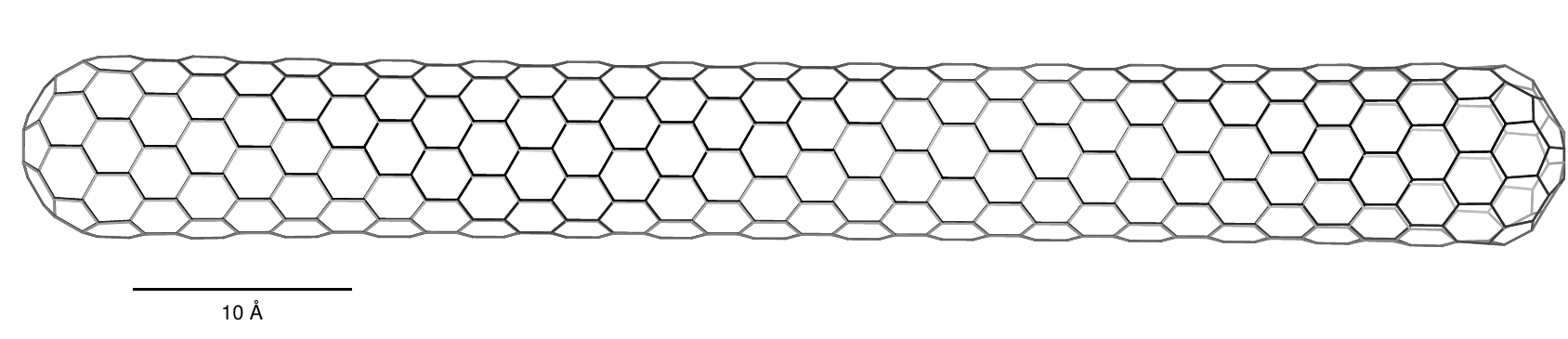}
    \end{subfigure}
    \hfill
    \begin{subfigure}[b]{0.3\textwidth}
        \centering
        \includegraphics[width=\textwidth]{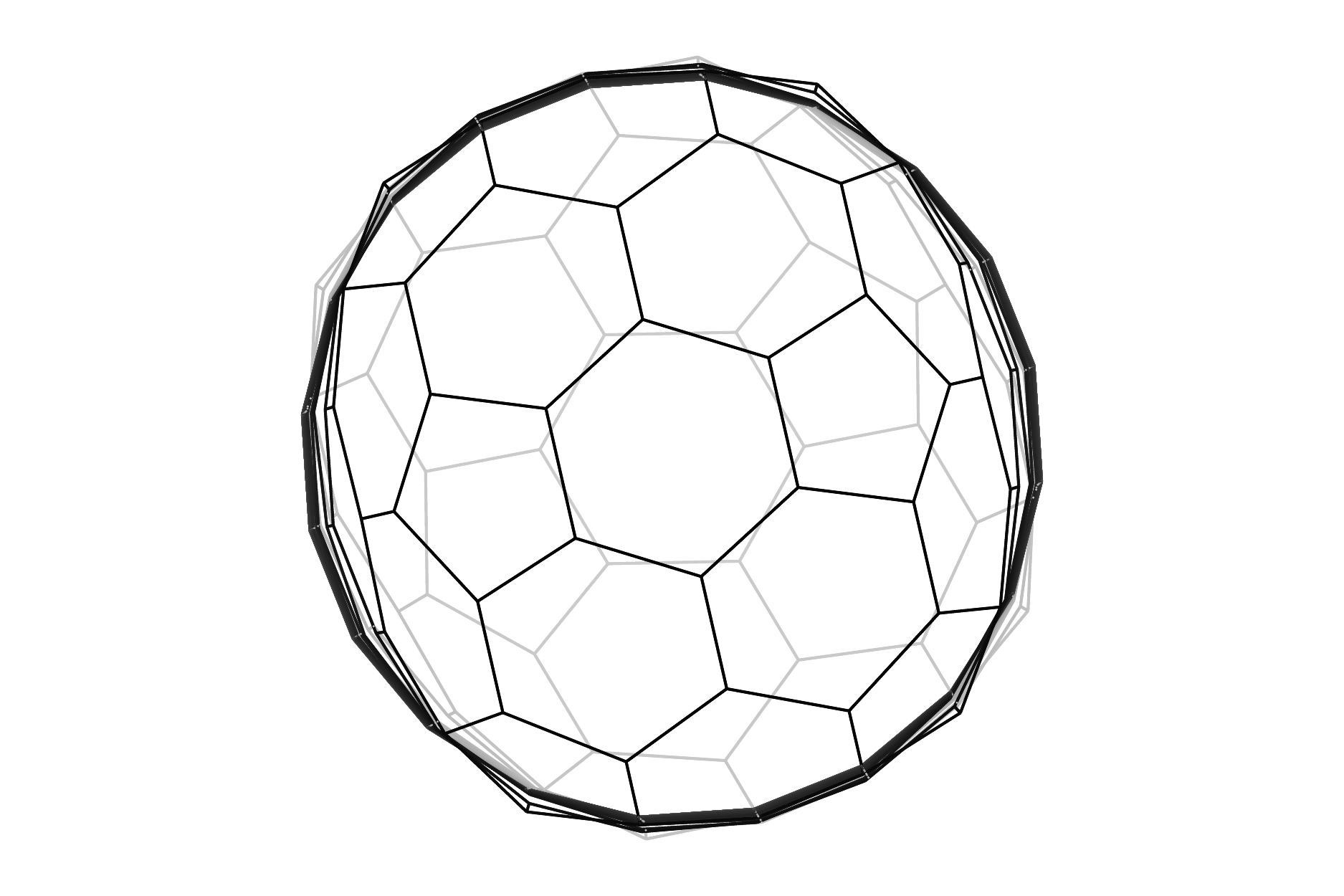}
    \end{subfigure}

        \setlength{\unitlength}{1mm}
          \begin{picture}(0,0)
            \put(-69.5,7.5){\rule{14.35mm}{2pt}}
            \put(-66,3.2){\colorbox{white}{\small10\,\AA}}
          \end{picture}

        \caption{A numerically validated approximation of a local energy minimizer for a \((10,0)\)-zigzag nanotube with $D_2$-symmetric caps, using the harmonic potential $\Eharmonic$.
        }
    \label{fig:nanotube_10_0_combined_hexagonal}
\end{figure}

The $(10,0)$-zigzag nanotube can be fitted with an alternative cap that has a hexagon at its center and dihedral $D_2$-symmetry \cite{LAIR2006447}. \Cref{fig:nanotube_10_0_combined_hexagonal} shows a numerical approximation of the $(10,0)$-zigzag nanotube with these caps. This numerical approximation is validated analogously to \Cref{theorem:existence_armchair_harmonic,theorem:existence_zigzag_harmonic}. To illustrate another application of our method, we provide rigorous measurements of the energy. As with all other numerical quantities, the energy evaluation produces a thin interval containing the true value, so we report the validated digits (those that stay the same over the full interval). \Cref{table:chirality_energies} lists the rigorous energy per atom for the five tube geometries presented.
As expected, one observes that the energy per atom decreases with increasing atom number since the influence of the caps decreases.
We also see that the $(10,0)$-zigzag nanotube is energetically more favorable than the $(5,5)$-armchair nanotube.

\begin{table}[ht]
    \centering
    \begin{tabular}{|c |c |c| c|}
        \hline
        Chirality & Cap symmetry       & Atoms & Energy per atom       \\
        \hline
        (5,5)     & Pentagonal & 670   &  0.31100507 \\
                \hline
        (5,5)     & Pentagonal & 370   &  0.50733624 \\
        \hline
        (5,5)     & Pentagonal & 360   &  0.51951604 \\
        \hline
        (10,0)    & Pentagonal & 660   &  0.28820774  \\
        \hline
        (10,0)    & $D_2$ & 364   &  0.29632260  \\
        \hline
    \end{tabular}
    \caption{Energy per atom for different chiralities using the harmonic potential.}
    \label{table:chirality_energies}
\end{table}

Besides other tube chiralities, we can analyze additional carbon structures, such as fullertubes and cyclocarbons~\cite{Koenig_2020,Stevenson_2024,Kaiser_2019}. We can also incorporate periodic offset conditions to model an infinite tube. In this setting, the nanotube is represented as a structure of atoms with boundary conditions, where atoms on the bottom boundary interact with those on the top. The difference is that these periodic boundary conditions remove one rotational symmetry so that the validation map \eqref{eq:validation_map} needs to be modified.

More general forms of energy have been used to model interaction energies, for example the bond order potentials \cite{PhysRevB.37.6991, Brenner1990, Pettifor1999} such as Tersoff's potential
\begin{equation}
    \label{eq:energy_tersoff}
    \ETersoff(p) \bydef \frac{1}{2} \sum_{i\neq j} f_C(r_{ij}) \left[ A e^{-\lambda_1 r_{ij}} - b_{ij} B e^{-\lambda_2 r_{ij}} \right],
\end{equation}
where the indices \(i\) and \(j\) run over the atoms of the system and \(r_{ij}\) denotes the distance between atoms \(i\) and \(j\). The cutoff function is defined as
\[
    f_C(r)=
    \begin{cases}
        1,                                                                      & r < R-D,       \\
        \frac{1}{2}\left\{ 1 - \sin\left[ \frac{\pi(r-R)}{2D} \right] \right\}, & R-D < r < R+D, \\
        0,                                                                      & r > R+D,
    \end{cases}
\]
with \(R\) chosen so that only the first-neighbor shell of atoms is included for most structures of interest. The bond order function is given by
\[
    b_{ij} = \left( 1 + \beta^n \zeta_{ij}^n \right)^{-\frac{1}{2n}}
\qquad
\text{with}
\qquad
    \zeta_{ij} = \sum_{k\neq i,j} f_C(r_{ik})\, g(\theta_{ijk}),
\]
where \(\theta_{ijk}\) is the angle between bonds \(i\!-\!j\) and \(i\!-\!k\). The angular function is defined as
\[
    g(\theta) = 1 + \frac{c^2}{d^2} - \frac{c^2}{d^2 + \bigl(h-\cos\theta\bigr)^2}.
\]
The model parameters listed in \Cref{table:tersoff_parameters} come from \cite{PhysRevLett.61.2879}, alternative values are found in \cite{PhysRevB.81.205441}.

\begin{table}[h]
    \centering\hfill
    \begin{tabular}{|c|c|}
        \hline
        A           & 1393.6 eV                                     \\ \hline
        B           & 346.74 eV                                     \\ \hline
        $\lambda_1$ & 3.4879 \AA{}                                   \\ \hline
        $\lambda_2$ & 2.2119 \AA{}                                   \\ \hline
    \end{tabular}\hfill
    \begin{tabular}{|c|c|}
        \hline
        $\beta$     & $1.5724 \times 10^{-7}$                       \\ \hline
        n           & 0.72751                                       \\ \hline
        c           & 38049                                         \\ \hline
        d           & 4.3484                                        \\ \hline
    \end{tabular}\hfill
    \begin{tabular}{|c|c|}
        \hline
        h           & -0.57058                                      \\ \hline
        R           & 1.95 \AA{}                                     \\ \hline
        D           & 0.15 \AA{}                                     \\ \hline
                    &                                                \\ \hline
    \end{tabular}\hfill

    \caption{Tersoff potential parameter values for carbon taken from \cite{PhysRevLett.61.2879}.}
    \label{table:tersoff_parameters}
\end{table}
\begin{figure}[!ht]
\centering
\setlength\unitlength{.8\textwidth}
\begin{picture}(1,.45)(.05,0)
\put(.075,.3){\includegraphics[width=.975\unitlength,trim=200 745 200 100,clip]{images/5_5_signed_radius_small.pdf}}
\put(.05,.02){\includegraphics[width=.95\unitlength,trim=0 15 10 10,clip]{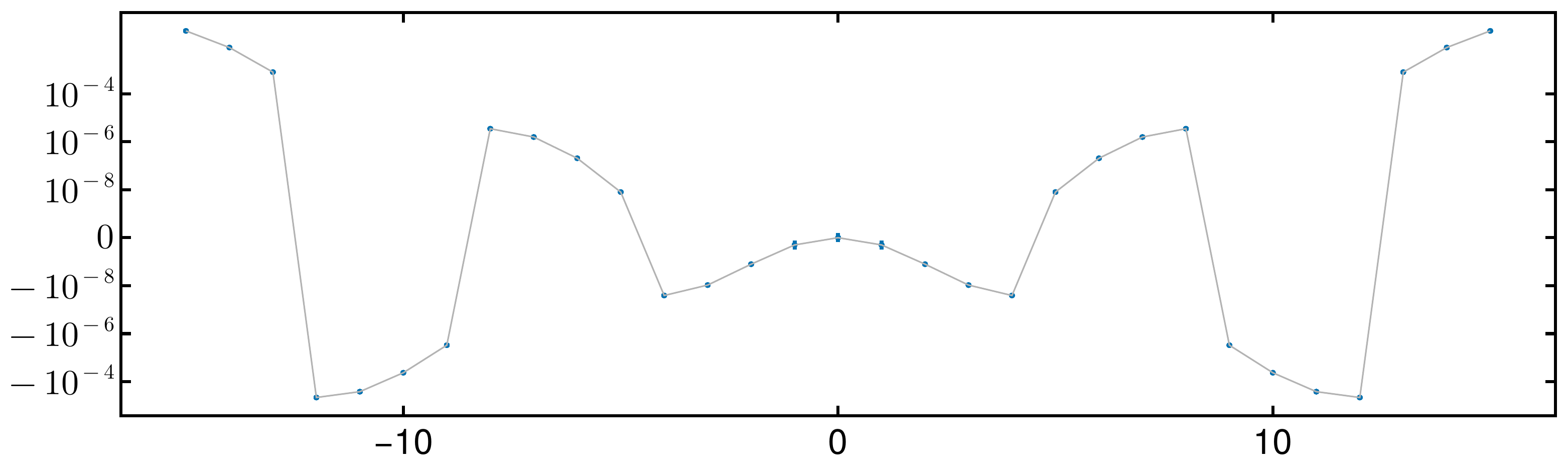}}
\put(0.42,0){\small Index of cross section}
\put(0,0.05){\rotatebox{90}{
  \begin{minipage}{0.25\unitlength}
    \centering \small
    Deviation from\\middle radius [\AA{}]
  \end{minipage}
}}
\end{picture}

   \caption{Same as \Cref{fig:nanotube_5_5_short_relative radii}, only replacing the harmonic potential \eqref{eq:energy_harmonic} by the Tersoff potential \eqref{eq:energy_tersoff}.}
    \label{fig:tersoff_radii}
\end{figure}

We use \Cref{theorem:minimizer} to validate numerical simulations using Tersoff's potential $\ETersoff$ with the original set of parameters from \cite{PhysRevLett.61.2879}.

\begin{theorem}[Capped armchair nanotube with Tersoff potential]
    \label{theorem:existence_armchair_Tersoff}
    Consider the approximate local minimizer $\bar{p}\in \mathbb{R}^{3n}$ of $\ETersoff$ from \eqref{eq:energy_tersoff} (with bonds corresponding to the (5,5)-armchair capped nanotube with $n=370$) shown in \Cref{fig:tersoff_radii}. There exists a local minimizer $p$ of $\ETersoff$ such that
    \[
        \|p - \bar{p}\|_\infty \le 4.6110\cdot10^{-11}\,\text{\AA}.
    \]
\end{theorem}

\begin{proof}

For this example, the proof differs only in the treatment of the cutoff function. Following Tersoff's original formulation \cite{PhysRevB.37.6991}, the cutoff radius \(R\) is chosen so that the cutoff selects only the first-neighbor shell for the structures of
interest. For a given candidate configuration, we validate this rigorously. 
More precisely, we compute interval bounds for all interatomic distances using interval arithmetic and then evaluate the Tersoff cutoff function on these distance intervals. This confirms that first-neighbor distances remain in the region where the cutoff is active and all other
distances remain outside the interaction range. In the implementation of the
proof we therefore iterate only over the nearest neighbors of each atom when evaluating
the energy and its derivatives.

The same approach applies when the cutoff includes more distant neighbors. We first validate the corresponding neighbor list for each particle and then restrict the evaluation of the energy and its derivatives to those interactions (as in the proof of \Cref{theorem:lattice_minimizer_PBC}). Using neighbor lists exploits the sparsity induced by the cutoff by restricting the interval arithmetic evaluation to the interacting pairs, which reduces computational cost. Moreover, it lets us use forward-mode automatic differentiation \cite{Tucker2011} to obtain interval versions of the gradient and Hessian of the energy.

Apart from the treatment of the cutoff, the detailed computational procedure of the proof follows the same structure as the one described after \Cref{theorem:minimizer}.
In this case, we obtain a true local minimizer of $\ETersoff$ that lies within distance $4.6110\cdot10^{-11}$ \AA{} of $\bar p$ and is unique within distance $9.9212\cdot10^{-8}$ \AA{} of $\bar p$ in the norm $\|\cdot\|_\infty$.

\end{proof}

As in the previous section we start by examining oscillations in the tube radius.
\Cref{fig:tersoff_radii} shows that for the Tersoff potential, just like for the harmonic potential (compare \Cref{fig:nanotube_5_5_short_relative radii}),
the distance of the atoms from the tube axis oscillates around a value with amplitude decaying exponentially fast away from the caps.
The magnitude of the oscillation is comparable to the harmonic potential, but the structure is different with groups of four consecutive cross sections having either a larger or a smaller distance.
The central tube radius is \(3.53221\,\text{\AA}\).

Another qualitative difference between energy minimizers of the Tersoff energy \eqref{eq:energy_tersoff} and the harmonic energy \eqref{eq:energy_harmonic} is
that the Tersoff energy minimizer exhibits multiple distinct bond lengths.
In particular, as \Cref{fig:tersoff_angles} shows, there are two preferred bond length values:
the bond lengths within each cross section (\(1.4759\)\,\AA) differ from those between two cross sections (\(1.4707\)\,\AA).
We also observe two distinct bond angles, \(2.07296\) and \(2.08030\) (both of which are smaller than $2\pi/3$). The smaller angle is adjacent to the bonds within a cross section, while the larger one occurs everywhere else.

\begin{figure}[!ht]
    \centering
    \includegraphics[width=\textwidth]{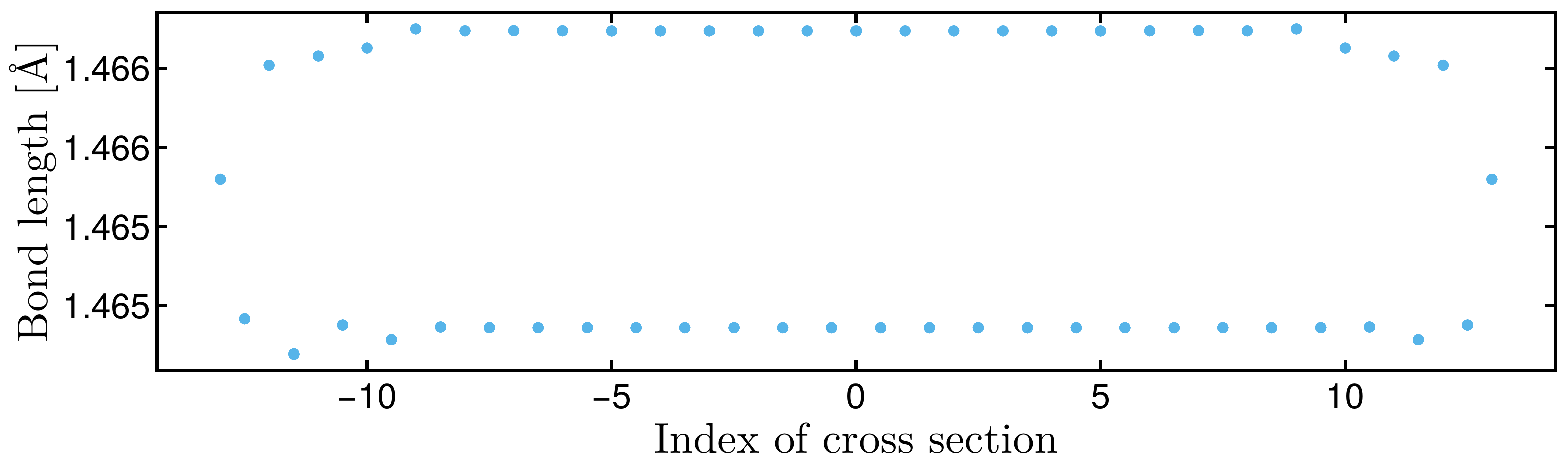}

    \includegraphics[width=\textwidth]{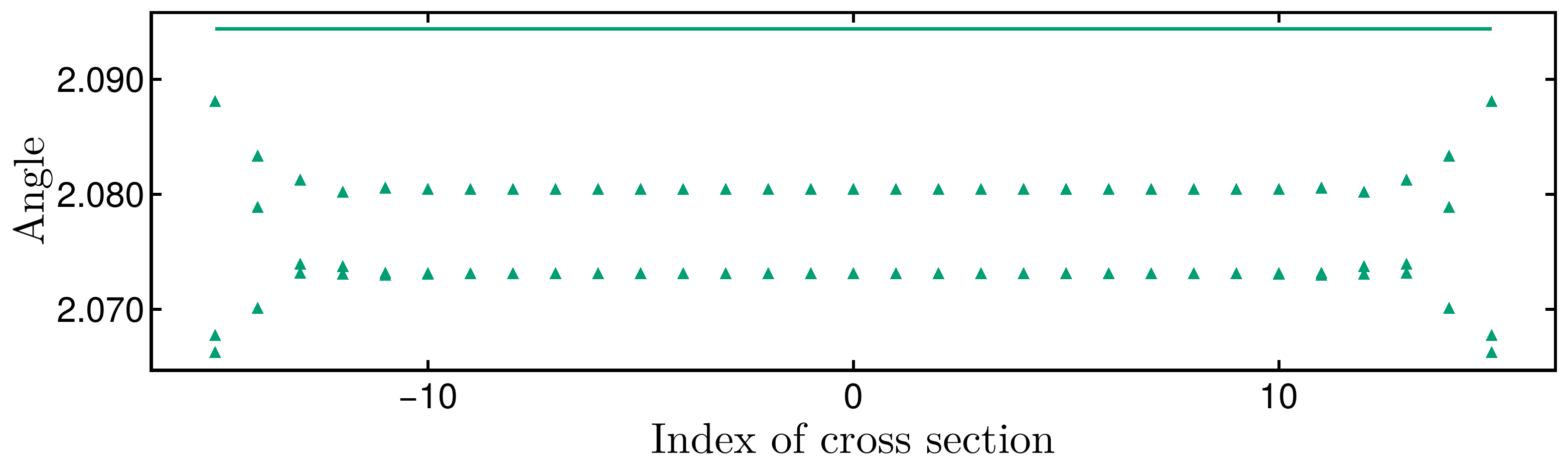}
    \setlength\unitlength\textwidth
    \begin{picture}(0,0)
    \put(-.96,.275){$2\pi/3$\,-}
    \end{picture}

\caption{Bond lengths (top) and bond angles (bottom) of the numerically validated approximation from \Cref{fig:tersoff_radii} of a local energy minimizer for a capped \((5,5)\)-armchair nanotube with 370 atoms, using the Tersoff potential.
Markers at noninteger indices correspond to bonds between two cross sections.}
    \label{fig:tersoff_angles}
\end{figure}
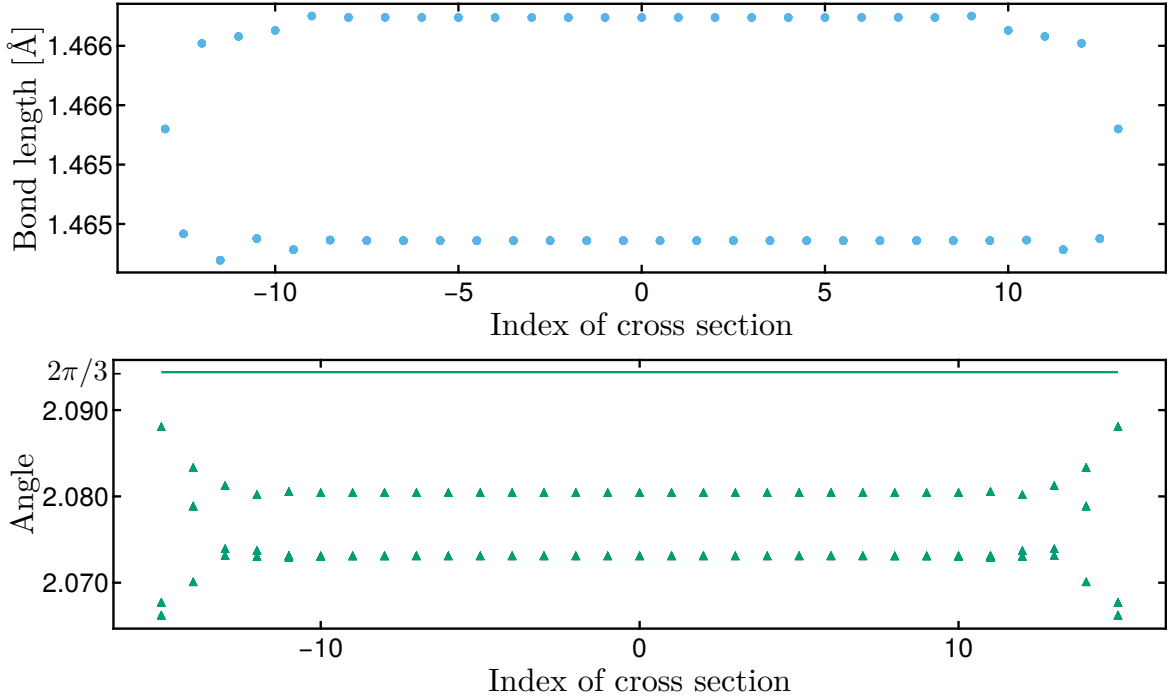

The Tersoff potential is an example of a parameterized function whose parameters are fit to measurement data. The choice of ansatz functions can lead here to qualitative differences. A potential of this type, with many parameters and composed functions, together with the high dimension of the problem, makes computer-based analysis necessary. A computer-assisted proof approach like the one we are presenting can be used to draw rigorous conclusions about local minimizers of the energy function using a numerical simulation.

As our final carbon nanotube example, consider the following non-physical Huber-norm potential
\begin{equation}
    \label{eq:energy_epsilon}
    \EHuber^\epsilon(p) \bydef  \sum_{ij \in B} \sqrt{\left(1 - r_{ij}/r_0 \right)^2 + \epsilon}
    \quad + \quad  \sum_{ij, jk \in B}
    \sqrt{\left(\tfrac{2\pi}{3} - \theta_{ijk} \right)^2 + \epsilon},
\end{equation}
which features a more pronounced minimum. Here the bond set \(B\) is defined analogously to the harmonic potential case.

The parameter \(\epsilon\) essentially controls the level of smoothness (as $\sqrt{a^2+\epsilon}$ tends to $|a|$ for $\epsilon\to0$). \Cref{fig:epsilon_radii} shows a numerically computed \((5,5)\)-armchair nanotube with \(n = 670\) atoms. It identifies oscillations in the cross-sectional radii that persist to a much greater distance from the caps than for the harmonic or the Tersoff potential. Again, the natural question raised by this observation is whether these oscillations are genuine features of the system or instead artifacts of the numerical computation. This too can be answered by rigorous validation.

\begin{theorem}[Capped armchair nanotube with Huber potential]
    \label{theorem:existence_armchair_Huber}
    Consider the approximate local minimizer $\bar{p}\in \mathbb{R}^{3n}$ of $\EHuber^\epsilon$ from \eqref{eq:energy_epsilon} for $\epsilon=10^{-5}$ (with bonds corresponding to the (5,5)-armchair capped nanotube with $n=670$) shown in \Cref{fig:epsilon_radii}. There exists a local minimizer $p$ of
    $\EHuber^\epsilon$ such that
    \[
        \|p - \bar{p}\|_\infty \le 1.8805\cdot10^{-37} .
    \]
\end{theorem}
\begin{proof}
For this final example, the validation follows the same procedure as in our previous CAPs. In this case, a true local minimizer of $\EHuber^\epsilon$ lies within distance $1.8805\cdot10^{-37}$ of $\bar p$ and is unique within distance $8.1854\cdot10^{-14}$ of $\bar p$ in the norm $\|\cdot\|_\infty$. The only difference is the numerical precision. For
simulations based on the energy function~\eqref{eq:energy_epsilon}, we use
quadruple precision to compute the numerical approximation, while the
validation step uses octuple precision. Most modern interval arithmetic libraries support multiple-precision floating
point arithmetic, so our method can take advantage of higher precision when
needed.
\end{proof}

\begin{figure}[!ht]
    \hspace*{.038\textwidth}\includegraphics[width=.963\textwidth,trim=160 720 257 190,clip]{images/5_5_signed_radius.pdf}
    \hspace*{.04\textwidth}\includegraphics[width=.96\textwidth,trim=0 30 0 30]{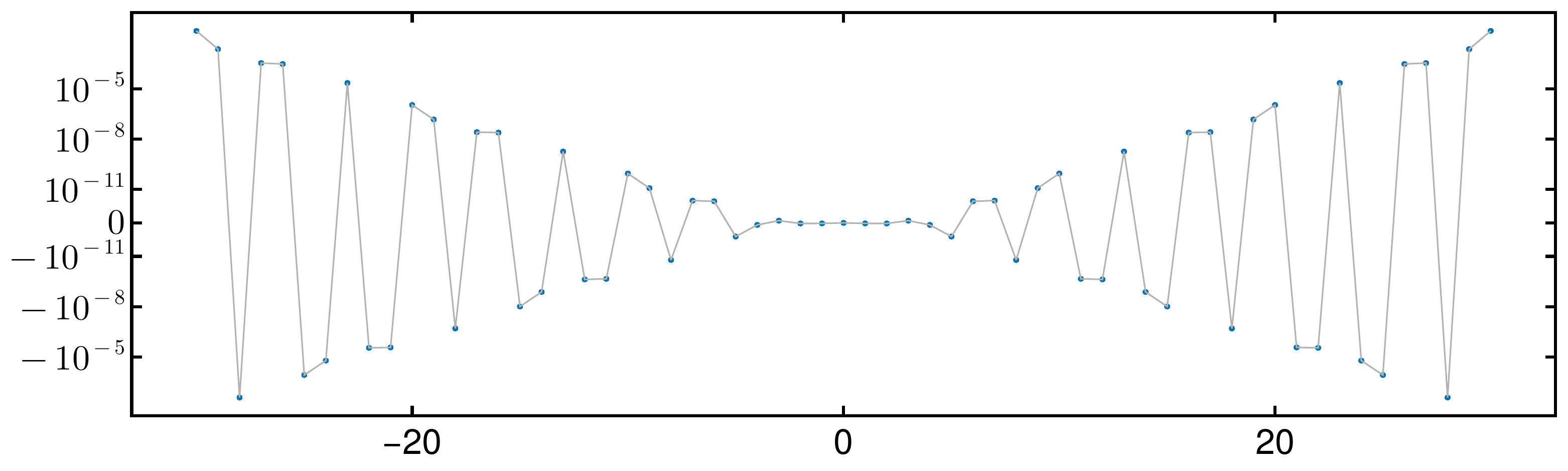}
   
\setlength\unitlength{\textwidth}
\begin{picture}(1,0)(0,0)

  \put(0.55, 0.02){
    \makebox(0,0)[t]{\colorbox{white}{ \strut Index of cross section}}
  }

  \put(-.02, 0.05){
    \rotatebox{90}{
      \colorbox{white}{
        \begin{minipage}{0.19\textwidth}
          \centering
          Deviation from\\middle radius 
        \end{minipage}
      }
    }
  }

\end{picture}

\caption{Same as \Cref{fig:nanotube_5_5_relative radii}, only replacing the harmonic potential \eqref{eq:energy_harmonic} by the Huber-type potential \eqref{eq:energy_epsilon} with $\epsilon=10^{-5}$.}
    \label{fig:epsilon_radii}
\end{figure}

Realistic physical potentials should not exhibit singularities such as discontinuities or nondifferentiabilities (natura non facit saltus),
which is why we employed rigorous numerical techniques that apply to differentiable energies.
However, it is conceivable that physical potentials can have different degrees of smoothness such as more or less pronounced minima,
which we emulate by varying the parameter $\epsilon$ in \eqref{eq:energy_epsilon}.
This does indeed have a qualitative influence on the geometry of the nanotube:
For the harmonic and Tersoff potentials, the nanotube radius quickly reaches a constant equilibrium value with increasing distance from the caps.
This does not happen for the less smooth potential with small $\epsilon$ (see \Cref{fig:epsilon_radii}); the smaller $\epsilon$, the longer-lived are the oscillations.
In fact, in the limit $\epsilon\to0$ it is expected that the radii of the transversal atom rings oscillate at a fixed amplitude along the full length of the capped nanotube.
Indeed, under the simplifying assumptions of fixed prescribed bond lengths and dihedral $D_5$-symmetry (properties that are exhibited anyway by our rigorously found nanotube configurations),
such behavior was shown in \cite[Thm.\,3]{Schrief22} for infinitely long \((5,5)\)-armchair nanotubes that contain an atom ring with perturbed radius.

\section{fcc Lattice}
\label{section:lattice}
Our second application comes from crystallization. We begin by addressing the role of boundary conditions. In our implementation, nanotube structures were simulated in free space, corresponding to open-boundary conditions. While this free-space, open-boundary setting is relevant, for example, when classifying and cataloging finite clusters under different interaction potentials (cf.\ \cite{Wales_2004, Wales_97, CambridgeClusterDatabase}), crystals are typically modeled as infinite lattices with imposed boundary conditions (PBCs) to avoid boundary effects and reduce computational cost. Our first examples are in a free-space framework. We initially anticipated that boundary effects would significantly influence the resulting crystallographic structures. However, contrary to this expectation, both our numerical simulations and supporting analytical results indicate that such effects are qualitatively negligible. This observation motivates the focus of the presentation that follows. For completeness, we conclude with an example incorporating periodic boundary conditions. 

We consider \(n\) atoms interacting through the Lennard--Jones potential and study geometry optimization for an fcc lattice configuration. For \(n\) atoms at positions \(p_1,\dots,p_n \in \mathbb{R}^3\), let
\(r_{ij} = \lVert p_i - p_j \rVert\) be the distance between atoms \(i\) and \(j\). We work with a nondimensional version of the Lennard--Jones energy per atom,
obtained by setting the parameters in the standard Lennard--Jones potential
equal to one. It is given by
\begin{equation}\label{eq:LJ-potential}
E_{\mathrm{LJ}}(p_1,\dots,p_n)
\bydef \frac{4}{n} \sum_{1 \le i < j \le n}
\left(
\frac{1}{r_{ij}^{12}}
-
\frac{1}{r_{ij}^{6}}
\right),
\end{equation}
where $r_{ij} = \|p_i - p_j\|$.

\begin{figure}[!ht]
  \centering
\begin{minipage}[b]{0.35\textwidth}
  \includegraphics[width=\linewidth,
    trim=0 12mm 0 12mm,clip]{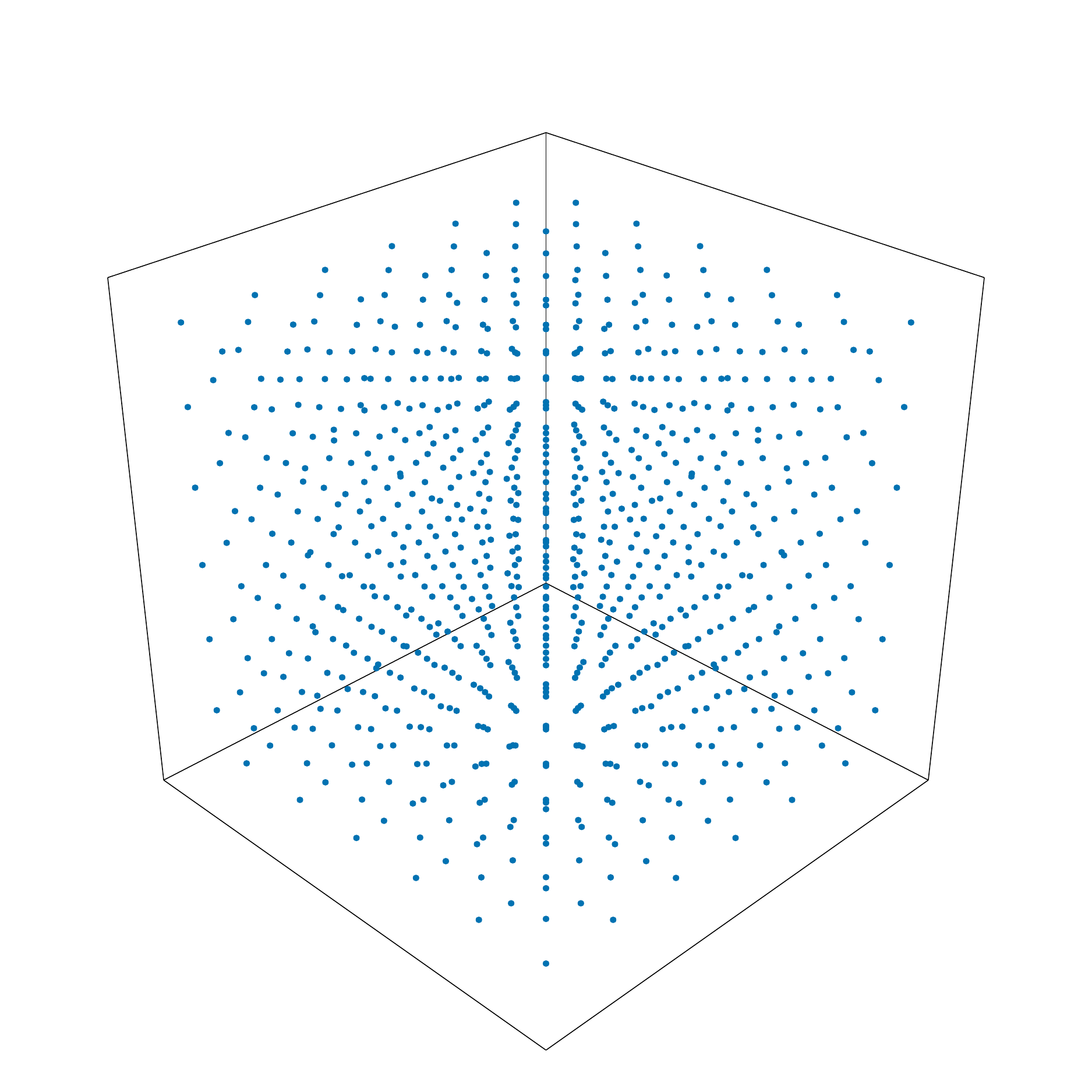}
\end{minipage}
\begin{minipage}[b]{0.25\textwidth}
  \includegraphics[width=\linewidth,
    trim=0 15mm 0 15mm,clip]{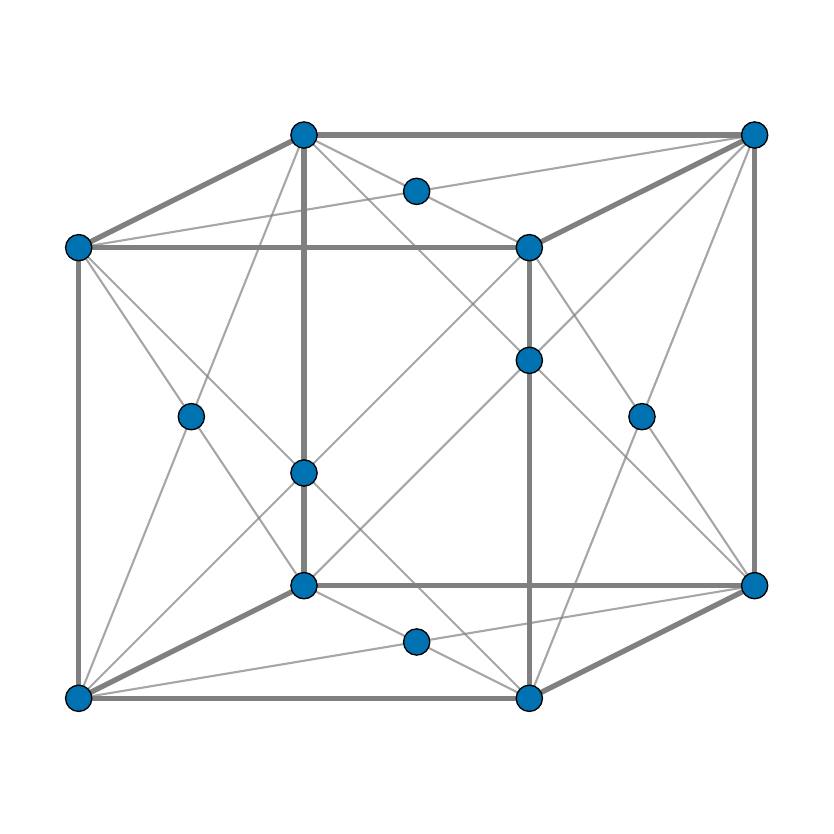}
\end{minipage}
  \caption{Left: approximate local minimizer of
the Lennard--Jones energy \(E_{\mathrm{LJ}}\) from \eqref{eq:LJ-potential},
corresponding to an fcc lattice block of \(864\) atoms. Right:
face-centered cubic (fcc) reference cell, with line segments drawn between
atoms to make the three-dimensional structure easier to see.}
    \label{fig:lattice_block}
\end{figure}
We begin with a numerically computed local minimizer that represents a perfect fcc crystal and, using our framework, obtain 
a computer-assisted proof that there exists a true local minimizer of the 
Lennard--Jones energy close to this configuration. Computer-assisted proofs of the theorems in this section follow the general
strategy described in Section~\ref{section:CAPs}. 
\begin{theorem}[fcc local minimizer]
    \label{theorem:lattice_minimizer}
Consider the approximate local minimizer \(\bar{p} \in \mathbb{R}^{3\cdot 864}\)
of the Lennard--Jones energy \(E_{\mathrm{LJ}}\) from \eqref{eq:LJ-potential}
corresponding to an fcc lattice block of \(864\) atoms 
shown in~\Cref{fig:lattice_block}. 
There exists a true local minimizer $p$ of
\(E_{\mathrm{LJ}}\) such that
    \[
        \|p - \bar{p}\|_\infty \le 1.4797\cdot 10^{-11},
    \]
and this
minimizer is unique within distance \(9.9994\cdot 10^{-8}\) of \(\bar{p}\) in
the norm \(\|\cdot\|_\infty\).
\end{theorem}
\begin{proof}
    The proof is analogous to that of Theorem~\ref{theorem:existence_armchair_harmonic}
and follows the same validation procedure.
\end{proof}

\begin{figure}[!ht]
  \centering
    \begin{minipage}[b]{0.25\textwidth}
    \includegraphics[width=\linewidth,clip]{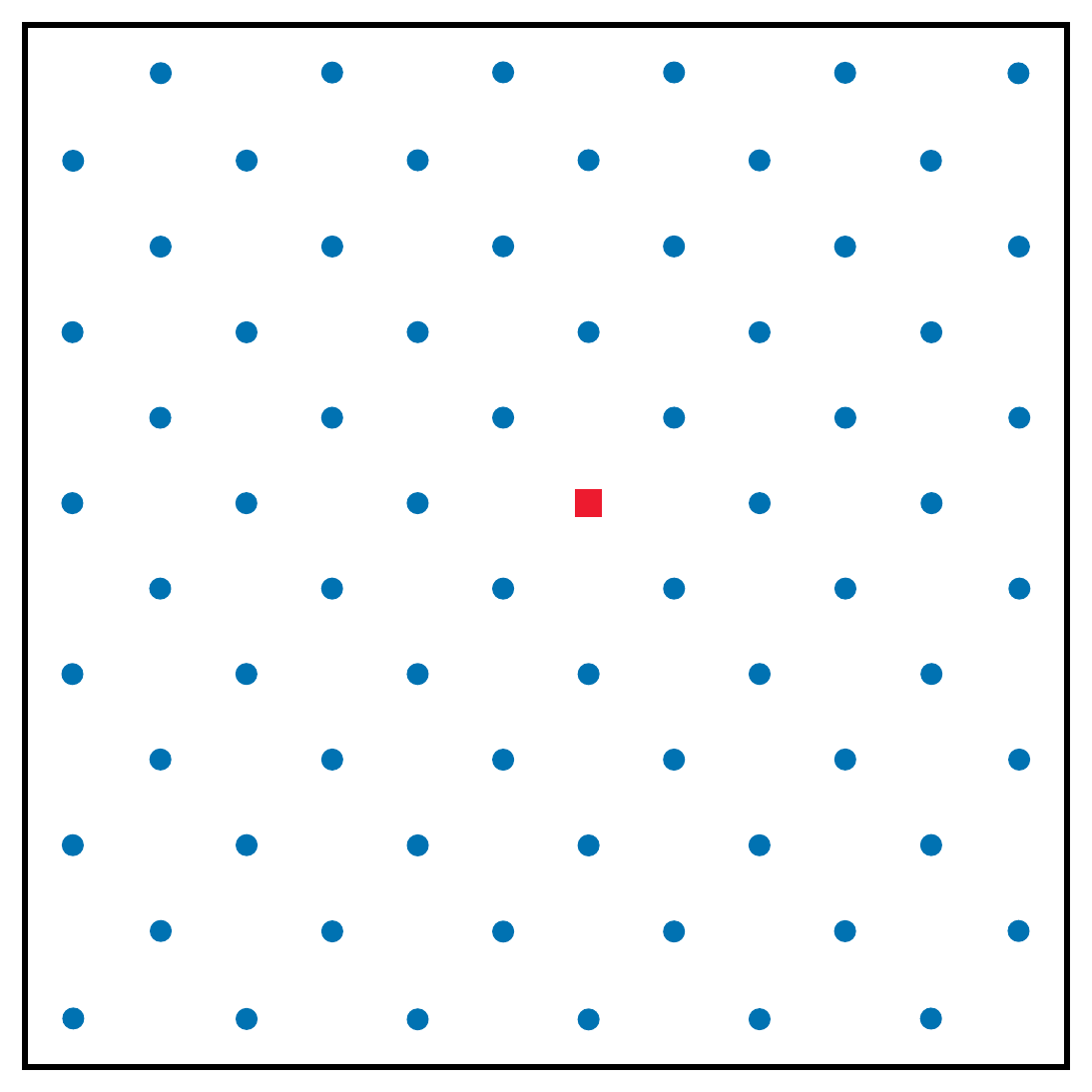}
  \end{minipage}
      \begin{minipage}[b]{0.7\textwidth}
    \centering
    \includegraphics[width=\linewidth]{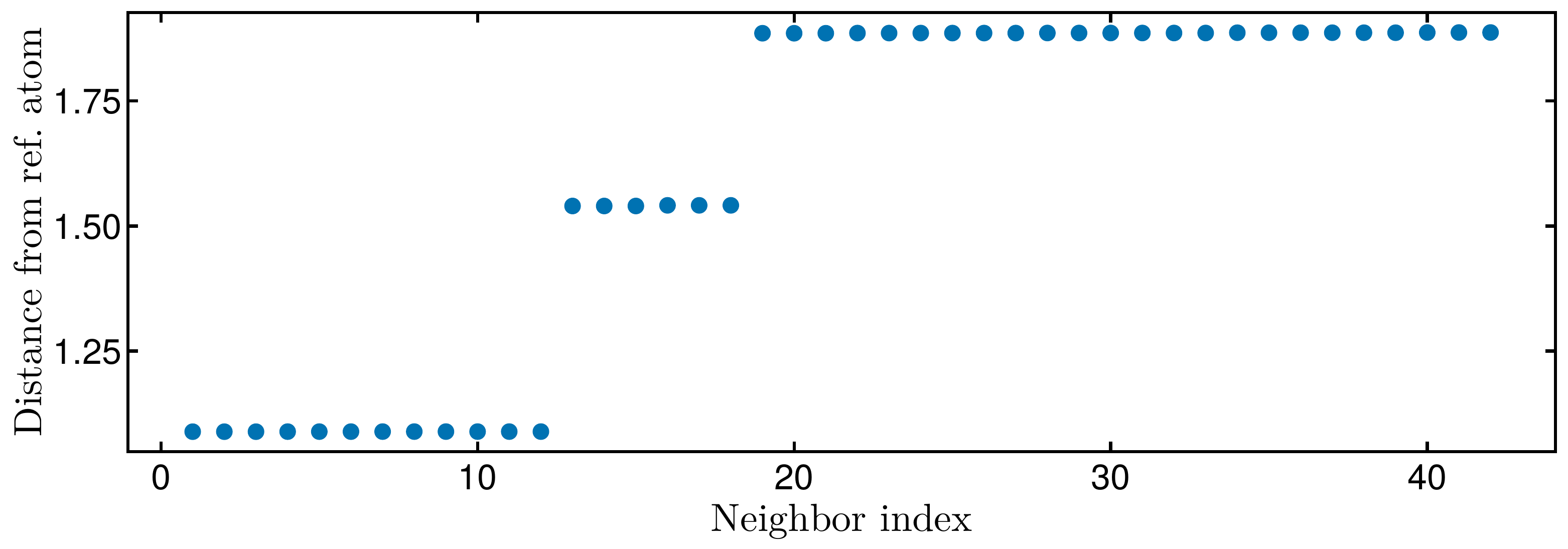}
  \end{minipage}

  \caption{Slice of the fcc lattice around a reference atom and neighborhood-distance
diagram. Left: slice of the lattice block in the local minimizer from
Theorem~\ref{theorem:lattice_minimizer}, with the reference atom indicated by a
red rectangle. Right: neighborhood-distance diagram for the reference atom.
Each point represents a neighboring atom. The vertical axis shows its distance
from the reference atom, and the horizontal axis orders atoms in the vicinity
by their distance from this reference atom. The diagram shows the fcc shell
structure in the first three shells, with 12 neighbors in the first shell, 6
in the second, and 24 in the third.}
    \label{fig:perfect_coordination}
\end{figure}
\begin{figure}[!ht]
  \centering
    \begin{minipage}[b]{0.25\textwidth}
    \includegraphics[width=\linewidth,clip]{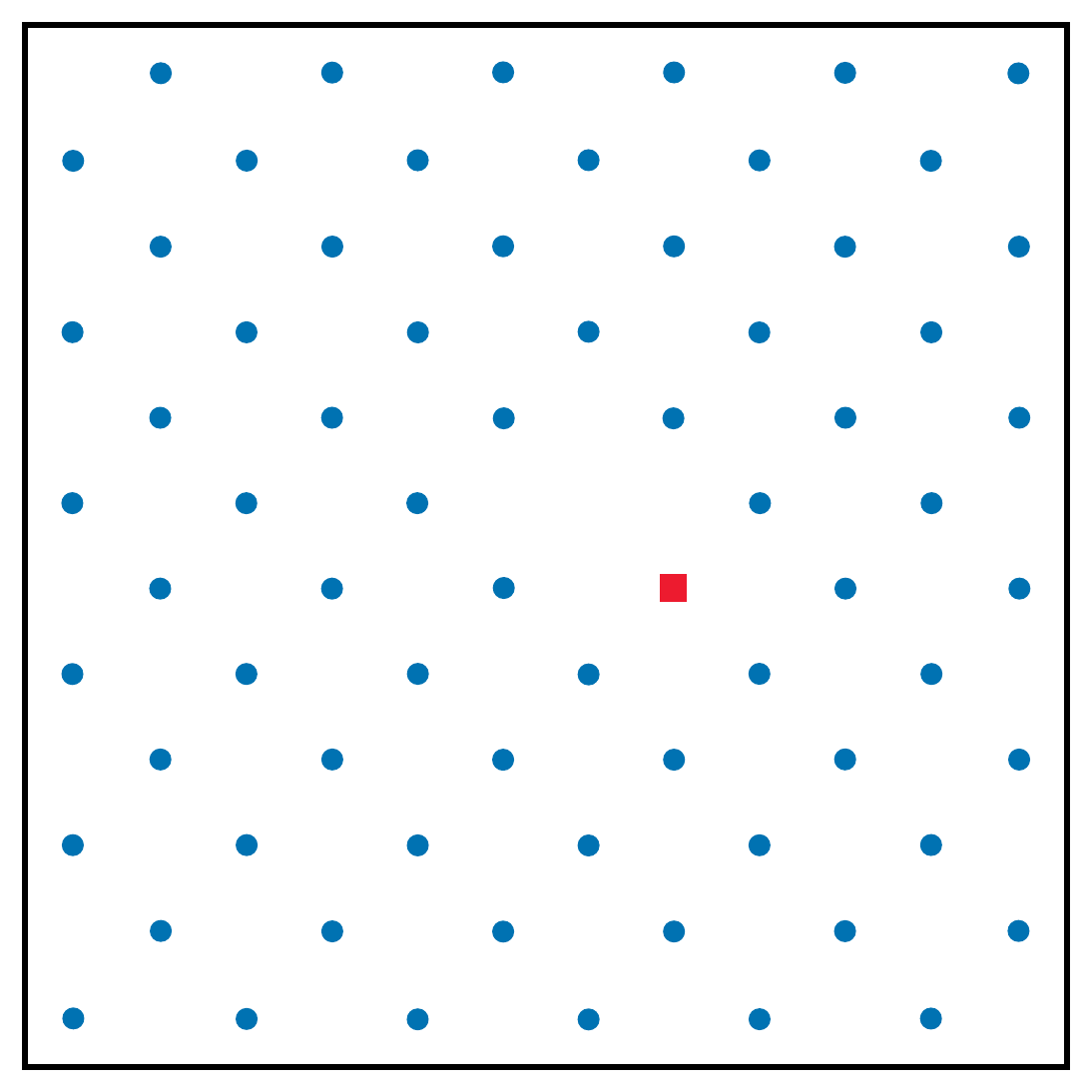}
  \end{minipage}
  \begin{minipage}[b]{0.7\textwidth}
    \centering
    \includegraphics[width=\linewidth]{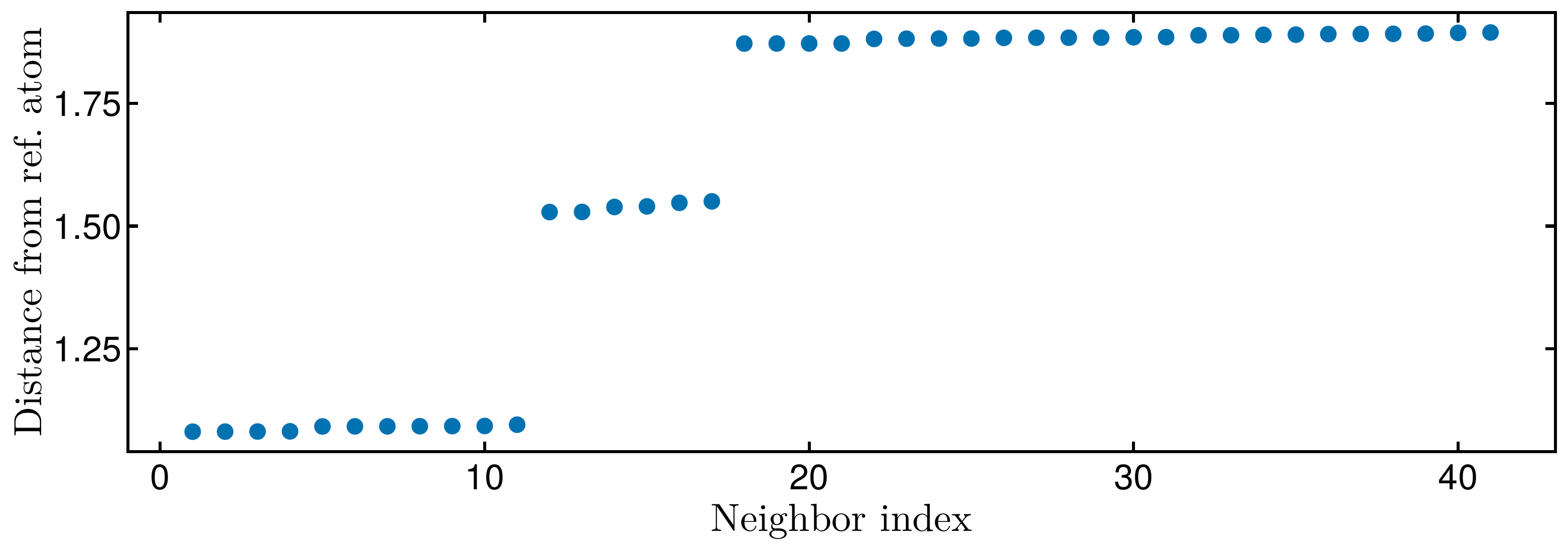}
  \end{minipage}

  \caption{Slice of the fcc lattice with a vacancy and neighborhood-distance diagram.
Left: slice of the lattice block in the local minimizer with a single vacancy defect from Theorem~\ref{theorem:lattice_vacancy}, with the reference atom indicated by a red rectangle and the vacant site visible next to it.
Right: neighborhood-distance diagram for the reference atom. Compared with the perfect fcc case, the first shell of the reference atom has one fewer neighbor and some atoms from the second shell move closer.}
    \label{fig:vacancy_coordination}
\end{figure}

\begin{figure}[!ht]
  \centering
  \begin{minipage}[b]{0.25\textwidth}
    \includegraphics[width=\linewidth,clip]{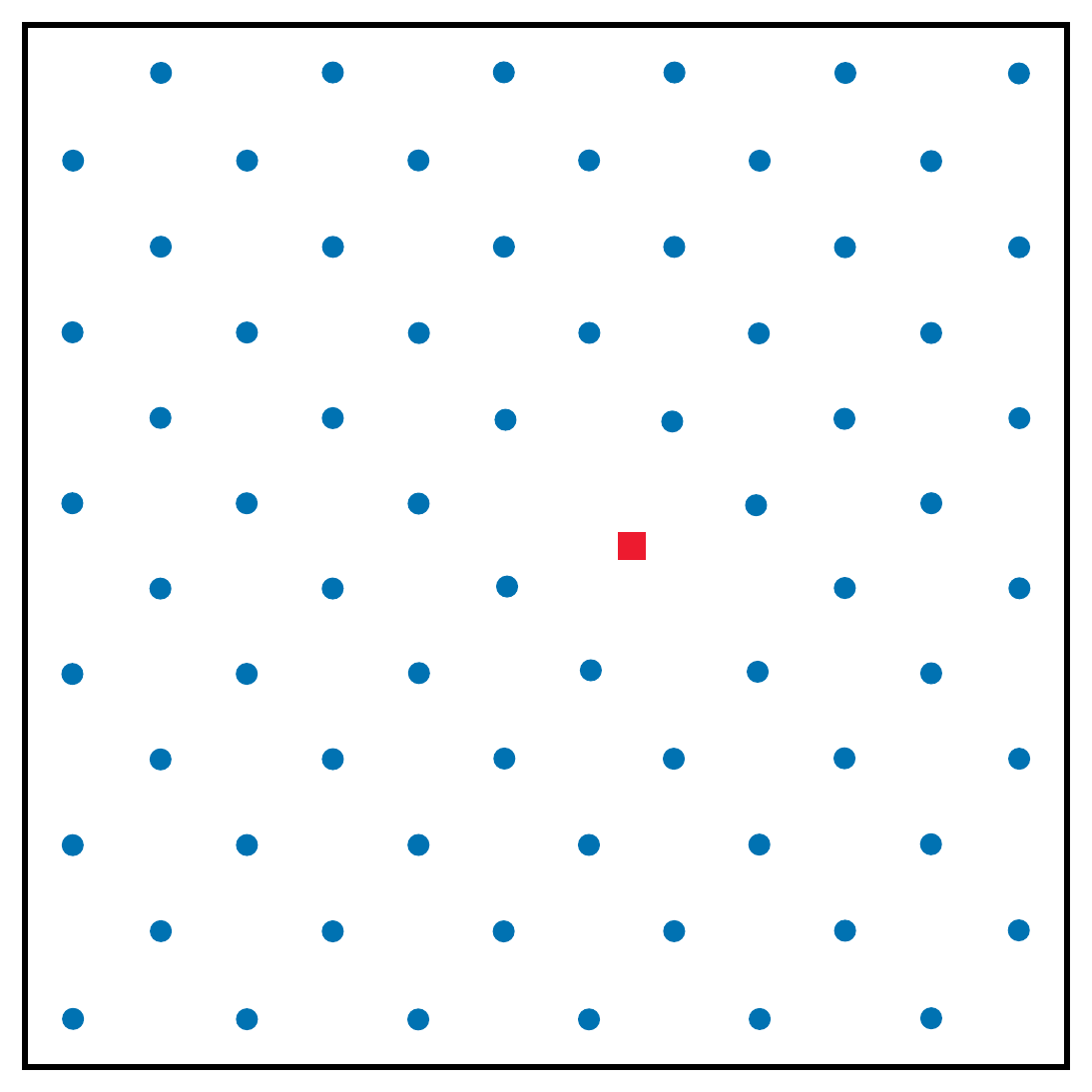}
  \end{minipage}
    \begin{minipage}[b]{0.7\textwidth}
    \centering
    \includegraphics[width=\linewidth]{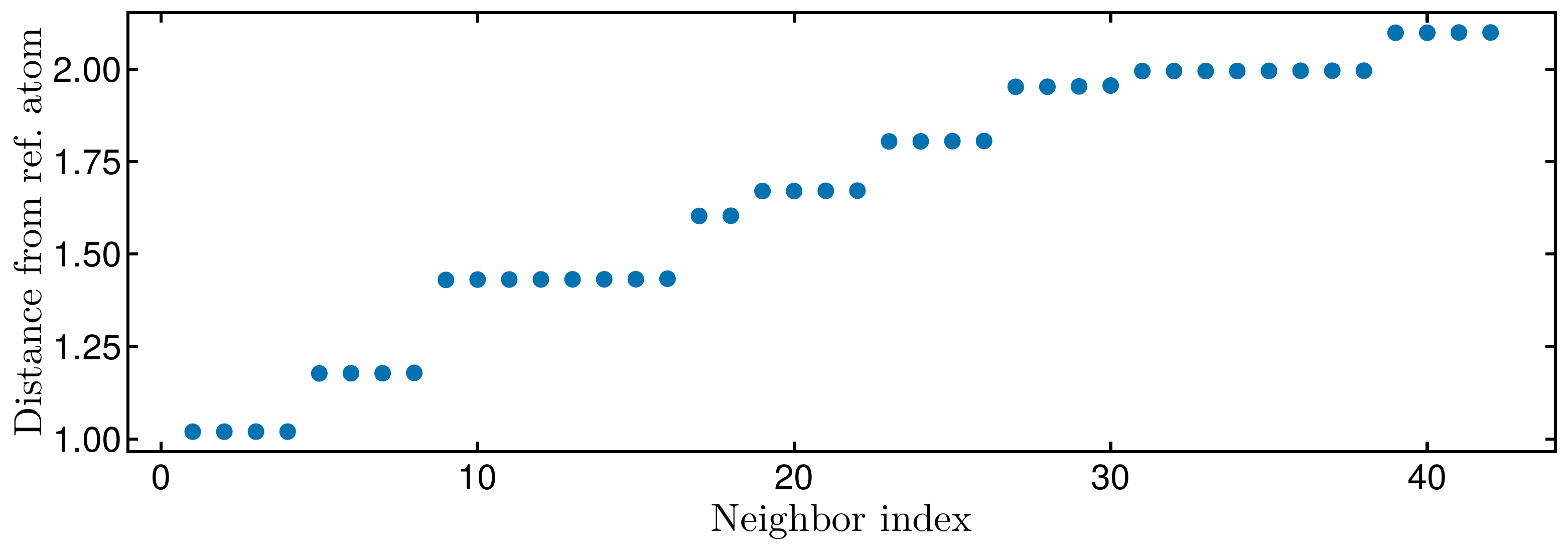}
  \end{minipage}
  \caption{Slice of the fcc lattice at the saddle and neighborhood-distance diagram.
Left: slice of the lattice block in the saddle configuration from
Theorem~\ref{theorem:lattice saddle}, with the reference atom indicated by a
red rectangle and the two vacancy sites visible on opposite sides of it.
Right: neighborhood-distance diagram for the reference atom. The fcc shell structure around the reference
atom is lost in this configuration.}
    \label{fig:saddle_coordination}
\end{figure}

Starting from the local minimizer of the perfect fcc lattice in
Theorem~\ref{theorem:lattice_minimizer}, we next create a defective
configuration by removing the atom whose position is closest to the center of
the lattice block. After numerical relaxation
this gives an approximate local minimizer with one vacancy, and we validate
this configuration using our method. At the scale of the
whole lattice block, this defective configuration is visually indistinguishable from the perfect lattice in \Cref{fig:lattice_block}, which is why we do not display it
separately. Instead, in Figure~\ref{fig:vacancy_coordination} we show a slice
of the block around the defect.
\begin{theorem}[fcc vacancy defect]
    \label{theorem:lattice_vacancy}
Consider the approximate local minimizer \(\bar{p} \in \mathbb{R}^{3\cdot 863}\)
of the Lennard--Jones energy \(E_{\mathrm{LJ}}\) from \eqref{eq:LJ-potential},
corresponding to the fcc lattice block with a single vacancy defect. A slice of
this lattice block where the defect can be seen is shown
in~\Cref{fig:vacancy_coordination}.
There exists a true local minimizer $p$ of
\(E_{\mathrm{LJ}}\) such that
    \[
        \|p - \bar{p}\|_\infty \le 7.9293\cdot 10^{-12},
    \]
and this
minimizer is unique within distance \(9.9994\cdot 10^{-8}\) of \(\bar{p}\) in
the norm \(\|\cdot\|_\infty\).

\end{theorem}
\begin{proof}
    The proof is analogous to that of Theorem~\ref{theorem:existence_armchair_harmonic}
and follows the same validation procedure.
\end{proof}

We then construct a second defective configuration as follows. Starting from
the local minimizer of the perfect fcc lattice, we identify the atom that is
the nearest neighbor of the atom removed to create the first vacancy and
remove this neighbor instead. After numerical relaxation, this configuration
is an approximate local minimizer with a vacancy on the neighboring site, and
we validate it analogously to Theorem~\ref{theorem:lattice_vacancy}. Starting
from these two validated local minimizers, we use the nudged elastic band (NEB) 
method \cite{Wales_2004} to compute an approximate saddle point on the Lennard--Jones energy
landscape. We then apply our computer-assisted
method to validate this saddle point as a true saddle. In fact, this represents a natural application for computer-assisted validation: In NEB computations, one simultaneously minimizes the energy of a number of states plus the energy of a sequence of fictitious springs that connect the states and thereby model an efficient transition path between the two end states. Since the spring constants and the path discretization are often chosen heuristically, rigorous validation could help distinguish converged transition states from artifacts of these choices.

\begin{theorem}[fcc saddle point]
    \label{theorem:lattice saddle}
Consider the approximate saddle point \(\bar{p} \in \mathbb{R}^{3\cdot 863}\)
of the Lennard--Jones energy \(E_{\mathrm{LJ}}\) from \eqref{eq:LJ-potential},
corresponding to a transition between the two vacancy configurations.
A slice of this configuration is shown in Figure~\ref{fig:saddle_coordination},
where the reference atom lies between the two vacancy sites. 
There exists a true saddle point $p$ of
\(E_{\mathrm{LJ}}\) such that
    \[
        \|p - \bar{p}\|_\infty \le 1.4409\cdot 10^{-10},
    \]
and this
saddle point is unique within distance \(9.9994\cdot 10^{-8}\) of \(\bar{p}\) in
the norm \(\|\cdot\|_\infty\).
\end{theorem}
\begin{proof}

The existence proof is analogous to that of Theorem~\ref{theorem:existence_armchair_harmonic}
and follows the same validation procedure. To prove that the validated critical point \(p\) is a saddle point, it suffices to show that the Hessian of \eqref{eq:LJ-potential} has a negative eigenvalue. For this, we consider the modified Hessian \(H(p)\) from \Cref{theorem:minimizer}, where the translational and rotational zero eigenvalues have been removed. Let \(P\) denote the closed ball of radius \(1.4409\cdot 10^{-10}\) centered at \(\bar p\) in the supremum norm. We compute an approximate eigenvector matrix \(Q\) for \(H(\bar p)\) and use interval arithmetic to rigorously evaluate the interval matrix
\[
Q^{-1} H(P)\,Q .
\]
For each \(q \in P\), the matrix \(Q^{-1}H(q)Q\) has the same eigenvalues as \(H(q)\). Since \(Q^{-1}H(P)Q\) is nearly diagonal, Gershgorin’s theorem \cite{GolubVanLoan2013} yields tight bounds on the spectrum. In particular, one of the Gershgorin discs lies strictly in the negative half-line and is disjoint from all the others. Hence \(H(q)\) has a negative eigenvalue for all \(q \in P\). Since the validated critical point \(p\) lies in \(P\), it follows that \(H(p)\) has a negative eigenvalue.
\end{proof}

In the shell plots presented in
\Cref{fig:perfect_coordination,fig:vacancy_coordination,fig:saddle_coordination}, we show neighborhood-distance diagrams with error bars given by the proofs for the
corresponding critical points of the energy. The error bars are so small that
they are visually indistinguishable from the plotted points. As in the nanotube
examples, every visible feature in these plots corresponds to a quantitative
statement that we have proved. In the vacancy configuration, the first shell of
the reference atom has one fewer neighbor (the vacancy), and some atoms in the
second shell move noticeably closer to the reference atom. In the saddle
configuration, the fcc shell structure around this atom is lost, since the atom
lies between the two vacancy positions. In the same spirit, one could study
boundary effects by choosing reference atoms near the boundary.

When modeling an infinite lattice, periodic boundary conditions (PBCs) are the standard choice for avoiding boundary effects and reducing computational cost. For our final example, we impose PBCs using the minimum image convention with a non-bonded cutoff \cite{Leach2001MolecularModelling}. The energy is then given by
\begin{equation}\label{eq:LJ-potential_PBC}
E_{\mathrm{LJ}}^{\mathrm{pbc}}(p_1,\dots,p_n)
\bydef
\frac{4}{n}
\sum_{1\le i<j\le n}
\mathbf{1}_{\{r_{ij}^{\mathrm{pbc}}<r_c\}}
\left(
\frac{1}{(r_{ij}^{\mathrm{pbc}})^{12}}
-
\frac{1}{(r_{ij}^{\mathrm{pbc}})^{6}}
\right),
\end{equation}
where
\begin{equation}\label{eq:min_image_distance}
r_{ij}^{\mathrm{pbc}}
\bydef
\min_{(m_1,m_2,m_3)\in\mathbb{Z}^3}
\bigl\lVert
p_j-p_i+(m_1L,m_2L,m_3L)
\bigr\rVert.
\end{equation}
We use a box of side $L=6$ and a cutoff of $r_c=2.5$, as suggested in \cite{Leach2001MolecularModelling}. 
\begin{theorem}[fcc local minimizer PBC]
    \label{theorem:lattice_minimizer_PBC}
Consider the approximate local minimizer \(\bar{p} \in \mathbb{R}^{3\cdot 256}\) of the Lennard--Jones energy \(E_{\mathrm{LJ}}^{\mathrm{pbc}}\) from \eqref{eq:LJ-potential_PBC}, representing an fcc block of \(256\) atoms in a cubic cell of length \(6\) with cutoff \(r_c=2.5\) (the numerical approximation \(\bar{p}\) is available in the code repository \cite{github_codes}). There exists a true local minimizer $p$ of
\(E_{\mathrm{LJ}}^{\mathrm{pbc}}\) such that
    \[
        \|p - \bar{p}\|_\infty \le 2.19826 \cdot 10^{-14}.
    \]
\end{theorem}
\begin{proof}
    The treatment of periodic boundary conditions and the cutoff is analogous to the Tersoff case in \Cref{theorem:existence_armchair_Tersoff}. Starting from an approximate local minimizer, we consider a ball of radius \(r\) centered at the approximation and evaluate the interparticle distances under the minimum-image convention and the cutoff. This yields a neighbor list for each particle. Using these lists, we evaluate the gradient and Hessian of the energy. The remainder of the proof follows the validation procedure described after \Cref{theorem:minimizer}, with one modification to the extended zero-finding problem. Under periodic boundary conditions the rotational invariance is no longer present; we therefore drop the rotational constraints and work with an extended map that fixes only translations. In this case,
\(
    F_{\bar{p}}^{\mathrm{pbc}}: \mathbb{R}^{3} \times \mathbb{R}^{3n} \to \mathbb{R}^{3} \times \mathbb{R}^{3n}
\)
is given by
\begin{equation}
    F_{\bar{p}}^{\mathrm{pbc}}(\tau,p)
    \bydef
    \begin{bmatrix}
        T_{1} \cdot  \left( \bar{p} - p \right)
        \\
        T_{2} \cdot \left( \bar{p} - p \right)
        \\
        T_{3} \cdot \left( \bar{p} - p \right)
        \\
        \nabla E(p)
        +
        \sum_{i=1}^{3} \tau_i T_i
    \end{bmatrix}.
\end{equation}
Here $T_i \in \mathbb{R}^{3n}$ for $i=1,2,3$ denote the infinitesimal generators of translations. For this example, a true local minimizer of $E_{\mathrm{LJ}}^{\mathrm{pbc}}$ lies within distance $2.19826 \cdot 10^{-14}$ of $\bar p$ and is unique within distance $9.9994\cdot10^{-07}$ of $\bar p$ in the norm $\|\cdot\|_\infty$.
\end{proof}  

Although here we restrict our attention to truncated energies, the computer-assisted approach can be extended to the non-truncated setting, where the energy includes contributions from an infinite number of particles. For example, once a validated local minimizer has been obtained, a posteriori corrections to the truncated energy that account for the missing long-range contributions (see, for example, \cite{Leach2001MolecularModelling}) can be implemented directly with interval arithmetic.

Long-range interactions can also be included explicitly in the energy as part of the formula. For instance, lattice sums \cite{Stillinger_2001} can be added to the energy and its derivatives. Along the same lines, Ewald sum-type methods incorporate long-range interactions using the fast Fourier transform \cite{Essmann_1995, Isele_2012, Takahashi2018} and may become rigorous with a verified FFT implementation and a rigorous treatment of aliasing error \cite{MR3709329,MR4930548}.

\section{Conclusions, Limitations, and Challenges Ahead}
 \label{section:conclusions}
Geometry optimization is studied in both the mathematical and broader scientific communities, where the search for local minimizers and transition states more often than not relies on numerical simulations. We introduce an approach that takes these numerical approximations and turns them into rigorous existence results for true solutions. We have illustrated our approach in the context of capped nanotubes and fcc crystallization. While the results and presentation here are geared toward applied mathematicians, we also see potential for use in wider scientific communities \cite{Mehta_2015, Mehta_2014}. 

We conclude by discussing the current limitations of the method and the challenges that remain. First, there is a fundamental mathematical limitation. These tools are currently unable to access or confirm the ground state (the global minimizer). As with any numerical tool, the complexity and nonconvexity of the energy landscapes entail the need for well-prepared initial conditions that reflect the overall geometry of the final proven state.

For the practitioner, we see two important limitations of our approach. A first limitation lies in the adaptation and implementation of the interatomic potential to interval arithmetic. In particular, the functions defining the gradient and the Hessian of the energy must be implemented so that they can be evaluated using interval arithmetic. In some languages, such as Julia, a single function can be written so that it accepts both floating-point and interval arguments, but in many cases, this will require a special interval arithmetic implementation. The implementation of a CAP can be simplified by computing the Hessian matrix of the energy via automatic differentiation. Libraries for interval arithmetic, such as \textit{IntervalArithmetic.jl} in Julia~\cite{IntervalArithmetic.jl} and \textit{INTLAB} in MATLAB~\cite{rump_intlab}, support rigorous derivative computation through forward-mode differentiation~\cite{Tucker2011}. We expect this additional implementation effort of a CAP (making the energy and its derivatives compatible with interval arithmetic) will decrease as artificial intelligence becomes more common in scientific computing workflows.

At the same time, contemporary potentials used today are highly empirical and increasingly the result of machine learning algorithms \cite{10.1063/5.0158783,10.1063/5.0005084}. In these cases, the applicability of our method depends on the compatibility with interval arithmetic. It is possible that the algebraic structure and complexity of the potential influence the amount of overestimation in the resulting interval bounds, and this may require higher precision arithmetic or extended interval arithmetic techniques, such as affine arithmetic \cite{figueiredo2004affine}. 
 
A second limitation concerns the error bounds associated with quantities computed at a minimizer, which derive from the foundational results of \Cref{section:CAPs}. Although these bounds can be extremely small, their significance must be interpreted relative to the physical parameters of the potential and the quantities under consideration. In particular, the resulting error estimates may occur at a level of precision far beyond that of experimentally accessible or physically meaningful scales, especially given that the parameters of the potential themselves are known only to finite accuracy.
That said, as illustrated in several examples presented here, there are situations in which these error bounds attain physical relevance. In addition, they provide a systematic means of distinguishing nonconverged simulations---or more generally, spurious numerical outputs---from genuine local minimizers and can also be used to rigorously differentiate between distinct critical points. Consequently, while standard computational approaches may often suffice in practice, there remain cases in which a more refined and rigorous methodology can offer meaningful additional insight.

\section*{Acknowledgments}
The authors are deeply indebted to the referees whose detailed comments substantially improved the manuscript.
They would also like to thank Vincent Meunier for his comments and suggestions. MA and RC thank Jean-Philippe Lessard for several helpful discussions. 

BW and RC were supported by the Deutsche Forschungsgemeinschaft (DFG, German Research Foundation) under Germany’s Excellence Strategy EXC 2044-390685587, Mathematics Münster: Dynamics-Geometry-Structure.
BW further gratefully acknowledges support by the DFG, priority program SPP 2256, Grant WI 4654/2-2, and by the Department of Mathematics and Statistics at McGill University. RC acknowledges support of NSERC (Canada) via its Discovery Grants Program and the support of the University of Münster as a Münster Research Fellow.

\bibliographystyle{siamplain}
\bibliography{references}

\end{document}